\documentclass[reqno]{amsart}

\usepackage{amssymb,amsmath,amscd,amsthm}

\usepackage{mathrsfs}
\usepackage{latexsym}
\usepackage{graphicx}

\newtheoremstyle{myplain}{}{}{\it}
{0pt}{\scshape}{}{ }{\thmname{#1}\thmnumber{ #2}\thmnote{ (#3)}}
\newtheoremstyle{mydefinition}{}{}{}
{0pt}{\scshape}{}{ }{\thmname{#1}\thmnumber{ #2}\thmnote{ (#3)}}

\theoremstyle{myplain}

    \newtheorem{Def}{Definition}[section]
        \newtheorem{Lem}[Def]{Lemma}

        \newtheorem{theo}[Def]{Theorem}
        \newtheorem{prop}[Def]{Proposition}
        \newtheorem{rem}[Def]{Remark}
        \newtheorem{hyp}[Def]{Hypothesis}
        \newtheorem{cor}[Def]{Corollary}

\newcommand{\skp}[2]{\mbox{$\left\langle #1\, , \, #2\right\rangle$}}
\newcommand{\skpd}[2]{\mbox{$\left\langle #1\, ,\,#2\right\rangle_{\ell^2}$}}

\newcommand{\skpR}[2]{\mbox{$\left\langle #1\, ,\,#2\right\rangle_{L^2}$}}

\newcommand{\skpH}[2]{\mbox{$\left\langle #1\, ,\,#2\right\rangle_{{\mathscr H}_{j}}$}}

\newcommand{\natop}[2]{\genfrac{}{}{0pt}{}{#1}{#2}}

\DeclareMathOperator{\dist}{dist}
\DeclareMathOperator{\Span}{span}

\DeclareMathOperator{\supp}{supp}

\DeclareMathOperator{\spec}{spec}

\DeclareMathOperator{\Op}{Op}

\DeclareMathOperator{\diag}{diag}
\DeclareMathOperator{\SO}{SO}

\newcommand{\id}{\mathbf{1}}

\numberwithin{equation}{section}

\newcommand{\beqa}{\begin{eqnarray*}}
\newcommand{\eeqa}{\end{eqnarray*}}
\renewcommand{\hat}{\widehat}
\newcommand{\bauf}{\begin{itemize}}
\newcommand{\eauf}{\end{itemize}}
\newcommand{\ben}{\begin{enumerate}}
\newcommand{\een}{\end{enumerate}}
\newcommand{\ra}{\rightarrow}

\renewcommand{\O}{\Omega}
\newcommand{\ep}{\varepsilon}

\newcommand{\R}{{\mathbb R} }
\newcommand{\Z}{{\mathbb Z}}
\newcommand{\C}{{\mathbb C}}
\newcommand{\N}{{\mathbb N}}
\newcommand{\T}{{\mathbb T}}

\newcommand{\hN}{\tfrac{\mathbb N^*}{2}}
\newcommand{\hNnull}{\tfrac{\mathbb{N}}{2}}
\newcommand{\hZ}{\frac{\mathbb Z}{2}}
\newcommand{\disk}{(\varepsilon {\mathbb Z})^d}
\newcommand{\Hi}{\mathscr H}
\newcommand{\Ce}{\mathscr C}

\newcommand{\De}{\mathscr D}
\newcommand{\E}{{\mathcal E}}
\newcommand{\F}{{\mathcal F}}

\title[Agmon estimates for the difference of exact and approximate eigenfunctions]{Agmon estimates for the difference of exact and approximate Dirichlet eigenfunctions for difference operators}

\author{Markus Klein \and Elke Rosenberger}

\address{
Universit\"at Potsdam\\ Institut f\"ur
Mathematik \\ Am Neuen Palais 10\\ 14469 Potsdam}
\email{mklein@math.uni-potsdam.de, erosen@uni-potsdam.de}

\date{\today}

\keywords{Semi-classical Difference operator, tunneling,WKB-expansions, Dirichlet eigenfunctions, asymptotic expansion, 
multi-well potential, Agmon estimates}

\begin{document}

\begin{abstract}
We analyze a general class of difference operators $H_\ep = T_\ep
+ V_\ep$ on $\ell^2(\disk)$, where $V_\ep$ is a multi-well
potential and $\ep$ is a small parameter. We  construct approximate eigenfunctions in 
neighbourhoods of the different wells and give weighted $\ell^2$-estimates for the difference of these
and the exact eigenfunctions of the associated Dirichlet-operators.
\end{abstract}

\maketitle

\section{Introduction}

This paper is motivated by the aim to find complete asymptotic expansions for the tunneling effect between different wells of a difference 
operator with multi-well potential. More precisely, we investigate a rather general class of families of
difference operators $\left(H_\ep\right)_{\ep>0}$ on the Hilbert space $\ell^2(\disk)$, as the small parameter  $\ep>0$
tends to zero. The operator $H_\ep$ is given by
\begin{align} \label{Hepein}
H_\ep &= (T_\ep + V_\ep ),  \quad\text{where}\quad
T_\ep  = \sum_{\gamma\in\disk} a_\gamma \tau_\gamma ,\\
(\tau_\gamma u)(x) &= u(x+\gamma) \, ,\qquad \quad (a_\gamma u)(x) := a_\gamma(x; \ep) u(x) \quad \mbox{for} 
\quad x,\gamma\in\disk \label{agammataugamma}
\end{align}
and $V_\ep$ is a multiplication operator which in leading order is given by a multi-well potential
$V_0 \in \Ce^\infty (\R^d)$, i.e. $V_0$ has more than one non-degenerate minima.

By  the tunneling effect, the interaction between neighboring
potential wells leads to the fact that
the eigenvalues and eigenfunctions are different from those
of an operator with decoupled wells, which is realized by the direct sum of ``Dirichlet-operators''
situated at the several wells. Since the interaction is small, it can be treated as a perturbation of
the decoupled system.

In \cite{kleinro4}, we showed that it is possible to approximate the
eigenfunctions of the original Hamilton operator $H_\ep$ with respect to a fixed spectral interval
by the eigenfunctions of the
several Dirichlet operators situated at the different wells and we gave
a representation of $H_\ep$ with respect to a basis of Dirichlet-eigenfunctions. In \cite{kleinro3} we constructed
formal asymptotic expansions for the Dirichlet-eigenfunctions at the wells.

In this paper, we show that to these formal expansions we can associate $\Ce_0^\infty$-functions defined in a larger neighbourhood of the wells. 
Moreover, we 
prove weighted $\ell^2$-estimates for the difference of these WKB-functions and the Dirichlet eigenfunctions.
This will allow us to compute complete asymptotic expansions for the elements of the interaction matrix and
to explicitly obtain the leading order term of eigenvalue splitting in a forthcoming paper.\\

This paper is based on the thesis Rosenberger \cite{thesis}. It is the fifth in a series of papers (see Klein-Rosenberger \cite{kleinro}, \cite{kleinro2}, 
\cite{kleinro3}, \cite{kleinro4}); the aim is to develop an analytic approach
to the semiclassical eigenvalue problem and tunneling for $H_\ep$ which in detail and
precision is comparable  
to the well known analysis for the Schr\"odinger operator (see Simon \cite{Si1} and
Helffer-Sj\"ostrand \cite{hesjo}). Our motivation comes from
stochastic problems (see Bovier-Eckhoff-Gayrard-Klein \cite{begk1}, \cite{begk2}). A large class of
discrete Markov chains, analyzed in \cite{begk2} with probabilistic
techniques, falls into the framework of difference operators treated in this article.

We expect that results similar to this paper hold for WKB expansions associated with eigenfunctions
of generators of jump processes in $\R^d$,  see Klein-L\'eonard-Rosenberger \cite{KLR} for  first results in this direction.

\begin{hyp}\label{hyp1}
For some $\ep_0>0$,
\begin{enumerate}
\item  the coefficients $a_\gamma(x; \ep)$ in \eqref{agammataugamma} are
functions
\begin{equation}\label{agammafunk}
a: \disk \times \R^d \times (0,\ep_0] \ra \R\, , \qquad (\gamma, x,
\ep) \mapsto a_\gamma(x; \ep)\, ,
\end{equation}
satisfying the following conditions:
\ben
\item[(i)] They have an
expansion
\begin{equation}\label{agammaexp}
a_\gamma(x; \ep) = \sum_{k=0}^{N-1} \ep^k a_\gamma^{(k)}(x)  + R^{(N)}_\gamma (x; \ep)\, ,\qquad N\in\N^*\, ,
\end{equation}
where $a_\gamma \in\Ce^\infty(\R^d\times (0,\ep_0])$ and $a_\gamma^{(k)}\in\Ce^\infty(\R^d)$ for
all $\gamma\in\disk$ and $0\leq k \leq N-1$.
\item[(ii)] $\sum_{\gamma\in\disk} a_\gamma^{(0)}  = 0$ and $a_\gamma^{(0)}
\leq 0$ for $\gamma \neq 0$.
\item[(iii)] $a_\gamma(x; \ep) =
a_{-\gamma}(x+\gamma; \ep)$ for all $x \in \R^d, \gamma \in \disk, \ep\in (0,\ep_0]$.
\item[(iv)] For any $c>0$ and $\alpha\in\N^d$ there exists $C>0$ such that for $0\leq k\leq N-1$, uniformly
with respect to $x\in\R^d$ and $\ep\in (0,\ep_0]$,
\begin{equation}\label{abfallagamma}
\| \, e^{\frac{c|.|}{\ep}} \partial_x^\alpha a^{(k)}_.(x)\|_{\ell_\gamma^2(\disk)}\leq C \qquad\text{and} \qquad
\|\,
 e^{\frac{c|.|}{\ep}} \partial_x^\alpha R^{(N)}_.(x)\|_{\ell^2_\gamma(\disk)}
 \leq C\ep^N\; .
\end{equation}
\item[(v)]
$\Span \{\gamma\in\disk\,|\, a^{(0)}_\gamma(x) <0\}= \R^d$ for all $x\in\R^d$.
\een
\item
\ben
\item[(i)] for all $\ep\in (0,\ep_0]$, the potential energy $V_\ep$ is the restriction to $\disk$ of a
function $\hat{V}_\ep\in\Ce^\infty (\R^d, \R)$ which has an expansion
\begin{equation}\label{hatVell}
\hat{V}_{\ep}(x) = \sum_{\ell=0}^{N-1}\ep^\ell V_\ell(x) + R_{N}(x;\ep)  \, ,\qquad N\in\N^*\, ,
\end{equation}
where $V_\ell\in\Ce^\infty(\R^d)$, $R_{N}\in \Ce^\infty (\R^d\times (0,\ep_0])$ and
for any compact set $K\subset \R^d$
there exists a constant $C_K$ such that $\sup_{x\in K} |R_{N}(x;\ep)|\leq C_K \ep^{N}$.
\item[(ii)]
$V_\ep$ is polynomially bounded and there exist constants $R, C > 0$ such that
$V_\ep(x) > C$ for all $|x| \geq R$ and $\ep\in(0,\ep_0]$.
\item[(iii)]
$V_0(x)\geq 0$ and it takes the value $0$ only at a finite number of non-degenerate minima
$x_j,\; j\in \mathcal{C} =\{1,\ldots , r\}$,
which we call potential wells.
\een
\een
\end{hyp}

If $\T^d := \R^d/(2\pi)\Z^d$ denotes the $d$-dimensional torus and
$b\in \Ce^\infty\left(\R^d\times \T^d\times (0,\ep_0]\right)$ for some $\ep_0>0$,
a family of pseudo-differential operators $\Op_\ep^{\T}(b): {\mathcal K}\left(\disk\right) \longrightarrow
{\mathcal K}'\left(\disk\right)$ is defined by
\begin{equation}\label{psdo2}
\Op_\ep^{\T}(b)\, v(x) := (2\pi)^{-d} \sum_{y\in\disk}\int_{[-\pi,\pi]^d}
e^{\frac{i}{\ep}(y-x)\xi}
b(x,\xi;\ep)v(y) \, d\xi 
\end{equation}
where
\begin{equation}\label{kompaktge}
{\mathcal K}\left(\disk\right):=\{ u: \disk\rightarrow \C\; |\; u~\mbox{has compact support}\}
\end{equation}
and ${\mathcal K}'\left(\disk\right):= \{f: \disk\rightarrow \C\ \} $ is dual to
${\mathcal K}\left(\disk\right)$
by use of the scalar product $\skpd{u}{v}:= \sum_x \bar{u}(x)v(x)$ (for details and properties 
of such pseudo-differential operators see Klein-Rosenberger \cite{kleinro2}).

We remark that for $T_\ep$ defined in
\eqref{Hepein}, under the assumptions given in Hypothesis \ref{hyp1}, one has $T_\ep = \Op_\ep^{\T}(t)$  where
$t\in\Ce^\infty\left(\R^d\times\T^d\times (0,\ep_0]\right)$ is given by
\begin{equation}\label{talsexp}
t(x,\xi; \ep) = \sum_{\gamma\in\disk} a_\gamma (x; \ep) \exp
\left(-\tfrac{i}{\ep}\gamma\cdot\xi\right)\; .
\end{equation}
Here $t(x,\xi;\ep)$ is considered as a function on $\R^{2d}\times (0,\ep_0]$, which is
$2\pi$-periodic with respect to $\xi$. By condition (a)(iv) in Hypothesis \ref{hyp1}, 
the function $\xi\mapsto t(x, \xi; \ep)$ has
an analytic continuation to $\C^d$. 
Moreover for all $B>0$
\begin{equation}\label{agammasum} 
\sum_\gamma \left| a_\gamma(x; \ep)\right| e^{\frac{B|\gamma|}{\ep}} \leq C \qquad\text{and thus}\qquad
\sup_{x\in\R^d} |a_\gamma(x; \ep)| \leq C e^{-\frac{B|\gamma|}{\ep}}
\end{equation}
for some $C>0$ uniformly with respect to $x$ and $\ep$.
We further remark that (a)(iv) implies $\bigl|a_\gamma^{(k)}(x)- a_\gamma^{(k)}(x + h)\bigr|\leq C |h|$ for $0\leq k \leq N-1$
uniformly with respect to $\gamma\in\disk$ and $x,h\in\R^d$ and (a)(ii),(iii),(iv) imply that $T_\ep$ is symmetric and 
bounded and that for some $C_T>0$
\begin{equation}\label{Tvonunten}
 \skpd{u}{T_\ep u} \geq - C_T \ep \|u\|^2_{\ell^2}\;, \qquad u\in\ell^2(\disk)\; .
\end{equation}

Furthermore, we set
\begin{align}\label{texpand}
t(x,\xi;\ep) &= \sum_{k=0}^{N-1} \ep^k t_k (x,\xi) + \tilde{t}_N(x,\xi;\ep)\quad
\text{with}\\
t_k(x, \xi) &:= \sum_{\gamma\in\disk} a_\gamma^{(k)}(x) e^{-\frac{i}{\ep}\gamma\xi}\, ,
\qquad 0\leq k \leq N-1\,,\nonumber\\
\hat{t}_N(x, \xi; \ep) &:= \sum_{\gamma\in\disk} R_\gamma^{(N)}(x; \ep)
e^{-\frac{i}{\ep}\gamma\xi}\nonumber\; .
\end{align}
Thus, in leading order, the symbol of $H_\ep$ is $h_0:=t_0+V_0$.
Combining \eqref{agammaexp} and (a)(iii) shows that the $2\pi$-periodic function $\R^d\ni\xi\mapsto t_0(x,\xi)$ 
is even with respect to $\xi\mapsto -\xi$, i.e.,
$a_\gamma^{(0)}(x) = a_{-\gamma}^{(0)}(x)$ for all $x\in\R^d, \gamma\in\disk$ (see \cite{kleinro}, Lemma 1.2) and therefore
\begin{equation}\label{tcosh} 
t_0(x,\xi) = \sum_{\gamma\in\disk} a_\gamma^{(0)}(x) \cos \bigl(\tfrac{1}{\ep}\gamma\cdot\xi\bigr) \; . 
\end{equation}
At $\xi=0$, for fixed $x\in\R^d$ the function $t_0$ defined in \eqref{texpand} has by Hypothesis \ref{hyp1}(a)(ii) an expansion
\begin{equation}\label{kinen}
t_0(x,\xi) = \skp{\xi}{B(x)\xi} + \sum_{\natop{|\alpha|=2n}{n\geq2}} B_\alpha (x) \xi^\alpha \qquad\text{as}\;\; |\xi|\to 0
\end{equation}
where $\alpha\in\N^d$, $B\in\Ce^\infty (\R^d, \mathcal{M}(d\times d,\R))$, for any $x\in\R^d$ the matrix $B(x)$  is
positive definite and symmetric and $B_\alpha$ are real functions. By straightforward calculations one gets
for $1\leq \mu,\nu\leq d$
\begin{equation}\label{Bnumu}
B_{\nu\mu}(x) = -\frac{1}{2\ep^2} \sum_{\gamma\in\disk} a_\gamma^{(0)}(x) \gamma_\nu\gamma_\mu\; .
\end{equation}

In order to work in the context of \cite{kleinro2}, we shall assume

\begin{hyp}\label{tildevarphihyp}
At the minima $x_j$ of $V_0,\, j\in\mathcal{C}$, we assume that $t_0$ defined in \eqref{texpand} fulfills
\[ t_0(x_j, \xi) >0  \quad\text{if} \quad |\xi| >0 \, .\]
\end{hyp}

We set, using \eqref{tcosh}, 
\begin{align}\label{tildehnull}
\tilde{h}_0\,:\,  \R^{2d} \ra \R\, , \quad \tilde{h}_0(x, \xi) &:= - h_0(x, i\xi) =: \tilde{t}_0(x,\xi)  - V_0(x)  \qquad \text{where }\\
\tilde{t}_0(x,\xi) &= - \sum_{\eta\in\Z^d} a^{(0)}_{\ep\eta}(x) \cosh \left(\eta\cdot \xi\right) \;.\nonumber
\end{align}
For any set $D\subset \R^d$, we denote the restriction to the lattice by  $D_\ep := D\cap \disk$.\\

The first step in this paper is the construction of quasimodes for $H_\ep$. 
For a one well operator, such  quasimodes were constructed near the potential well in \cite{kleinro3}, using the formal asymptotic
expansions obtained in that paper. Here we generalize  to our multiwell case and extend the quasimodes to larger domains, using
the Hamiltonian flow associated to $\tilde{h}_0$. In an adapted form, we restate some of the results of \cite{kleinro3} in the Appendix.

To this end, we consider $H_\ep$ as an operator on $L^2(\R^d)$, i.e. we construct quasimodes with respect to the
operator on $L^2(\R^d)$
\begin{equation}\label{hatHep}
 \widehat{H}_\ep = \widehat{T}_\ep + \widehat{V}_\ep\qquad \text{with}\quad 
\widehat{T}_\ep = \sum_{\gamma\in\disk} a_\gamma \tau_\gamma
\end{equation}
where $a_\gamma\tau_\gamma$ is the translation operator defined in \eqref{agammataugamma}.

The operator on $L^2(\R^d)$ associated to $h_0$ is given by
\begin{equation}\label{Hnullhut}
 \hat{H}_0  := \sum_\gamma a_\gamma^{(0)} \tau_\gamma + V_0\; .
\end{equation}
The harmonic oscillator at the potential wells $x_j,\, j\in \mathcal{C}$, associated to $\hat{H}_0$ is given by the operator on $L^2(\R^d)$
\begin{equation}\label{Hnullhutq}
 \hat{H}_{0,q}^j(x,\ep D) = -\ep^2 \skp{D}{B_j D} + \skp{x}{A_j x} + \ep \bigl(t_1(x_j, 0) + V_1(x_j)\bigr)
\end{equation}
where $A_j=D^2 V_0|_{x_j}$ and $B_j=B(x_j)$ for $B$ given in \eqref{Bnumu}.
We introduce the unitary transformation
\begin{equation}\label{unitrans}
  U_jf(x) = |\det B_j|^{-\frac{1}{4}} f(C_j x) \, , \quad C_j:=  R_j B_j^{-\frac{1}{2}} 
\end{equation}
on $L^2(\R^d)$, where $R_j\in \SO (d,\R)$,  such that $\Lambda_j:= R_jB_j^{\frac{1}{2}}A_j B_j^{\frac{1}{2}}R_j^T$ is diagonal. 
Then the operator 
\begin{equation}\label{hatHnullstrich}
 \hat{H}'_{0,j} := U_j^{-1} \tau_{-x_j}\hat{H}_0 \tau_{x_j} U_j \qquad\text{and}\quad 
\end{equation}
is centered at $0$
and, for $z= x-x_j$, the associated symbol $h_{0,j}(z,\zeta) = h_0(C_j^{-1} z, C_j^T \zeta)$ with respect
to the $\ep$-quantization given in \eqref{psdo2}
is in leading order diagonal: 
\[ h_{0,j}(z,\zeta) = \zeta^2 + \skp{z}{\Lambda z} + O(|z|^3) + O(|\zeta|^3)\, , \qquad \Lambda = \diag (\lambda_1,\ldots \lambda_n) \]
(the transformation of $h_0$ follows  from the usual change of variable formulae for pseudo-differential operators in view of
e.g. (A.4) in \cite{kleinro2}). Moreover, for $\widehat{H}_\ep$ and $U_j$ given in \eqref{hatHep} and \eqref{unitrans}, we set
\begin{equation}\label{Hstrichj}
\widehat{H}'_{\ep,j}:= U_j^{-1}\tau_{-x_j} \widehat{H}_\ep \tau_{x_j} U_j\; .
\end{equation}

By Hypothesis \ref{hyp1}, $\tilde{h}_0$ is hyperregular, and even and  strictly convex in each fibre (cf. \cite{kleinro}). We can thus
introduce  the associated Finsler distance $d = d_\ell$ on $\R^d$ as in \cite{kleinro}, Definition 2.16, where we set 
$\widetilde{M}:=\R^d\setminus\{x_k\, , \, k\in\mathcal{C}\}$.
Analog to \cite{kleinro}, Theorem 1.6, it can be shown that $d$ is locally Lipschitz and that for any $j\in\mathcal{C}$, the distance 
$d^j(x):= d(x, x_j)$ 
fulfills the generalized eikonal equation and inequality respectively 
\begin{align}\label{eikonal2} 
\tilde{h}_0 \bigl(x, \nabla d^j (x)\bigr) &= 0 \; ,\qquad x\in\Omega^j \\
\tilde{h}_0\bigl(x, \nabla d^j(x)\bigr) &\leq 0\; , \qquad x\in\R^d
\end{align}
where $\Omega^j$ is some neighborhood of $x_j$.

It follows directly that for each $j\in\mathcal{C}$
\begin{equation} \label{djhut}
 \hat{d}^j(z) := d^j \bigl(x_j + C_j^{-1}z\bigr)
\end{equation}
is the Finsler distance associated to $\tilde{h}_{0,j}$ defined in \eqref{eikonal}, 
satisfies the eikonal equation
\begin{equation}\label{eikonal}
\tilde{h}_{0,j} \bigl(z, \nabla \hat{d}^j (z)\bigr):= -h_{0,j}\bigl(z, i\nabla \hat{d}^j (z)\bigr) = 0\; ,
\qquad z\in C_j\bigl(\Omega^j - x_j\bigr) 
\end{equation}
where $\Omega^j$ is some neighborhood of $x_j$
and fulfills
\begin{equation}\label{djentw}
 \bigl|\hat{d}^j(z) - \hat{d}^j_0(z)\bigr| = O(|z|^3) \quad (z\to C_j x_j)\qquad \text{for}\quad 
\hat{d}^j_0(z) =\sum_{\nu=1}^d \frac{\lambda^j_\nu}{2} z_\nu^2 \, ,
\end{equation}
where $\lambda^j_\nu >0, \, \nu=1,\ldots d,$ are the eigenvalues of $A_j$.

We remark that, assuming only Hypothesis \ref{hyp1},
it is possible that balls $B_r(x):=\{y\in\R^d\,|\, d(x,y)\leq r\}, r<\infty$, are unbounded in the Euclidean distance (and thus
not compact). In this paper, we shall not discuss consequences of this effect.

\begin{rem}\label{remhypmultwell}
Since $d$ is locally  Lipschitz-continuous, it follows from \eqref{agammasum} that for any
$B>0$ and any bounded region $\Sigma\subset \R^d$ there exists a constant
$C>0$ such that
\begin{equation}\label{agammasupnorm2}
\sum_{\natop{\gamma\in\disk}{|\gamma|<B}} \Bigl\|a_\gamma (\,.\, ; \ep)
e^{\frac{d(.,.+\gamma)}{\ep}}\Bigr\|_{l^\infty(\Sigma)} \leq C\; .
\end{equation}
\end{rem}

For a global estimate on the decay of $a_\gamma$ for large $\gamma$, we assume in addition

\begin{hyp}\label{Hypo3}
There exist constants $\eta>0$ and $C>0$ such that for all $x\in\R^d$
\[ \left\| a_{(.)} (x; \ep) e^{\frac{1}{\ep} d(x, x+\, . \, )} |\, . \, 
|^{\frac{d + \eta}{2}} \right\|_{\ell^2(\disk)} \leq C \; . \]
\end{hyp}

For the WKB-expansion of the eigenfunctions constructed in \cite{kleinro3} and the weighted
norm-estimates for the Dirichlet eigenfunctions shown in \cite{kleinro}, it was essential
that the potential $V_0$ had exactly one minimum.
To use these results it is necessary to restrict $H_\ep$ to regions around the wells, which exclude all other wells. 

For $\Sigma\subset\R^d$ we define the space
$\ell^2_{\Sigma_\ep}:= i_{\Sigma_\ep} \left(\ell^2(\Sigma_\ep)\right) \subset \ell^2(\disk)$ where
$ i_{\Sigma_\ep}$ denotes the embedding via zero extension.
Then we define the Dirichlet operator
\begin{equation} \label{HepD}
H_\ep^{\Sigma} :=\id_{\Sigma_\ep} H_\ep|_{\ell^2_{\Sigma_\ep}}  \;:\; \ell^2_{\Sigma_\ep} \rightarrow \ell^2_{\Sigma_\ep}
\end{equation}
with domain $\De (H_\ep^{\Sigma}) = \{u\in\ell^2_{\Sigma_\ep}\,|\, V_\ep u \in \ell^2_{\Sigma_\ep}\}$.

\begin{hyp}\label{hypIMj}
\ben
\item For $j\in\mathcal{C}$, we choose a
compact manifold $M_j\subset \R^d$ with $\Ce^2$-boundary such that $x_j\in M_j$,  $d^j\in\mathscr{C}^1(M_j)$ and
 $x_k\notin M_j$ for $k\in\mathcal{C}, \,k\neq j$.
\item
Given $M_j,\,  j\in\mathcal{C}$, let $I_\ep = [\alpha (\ep),\beta (\ep)]$ be an interval, such that
$\alpha (\ep),\beta (\ep) \to 0$ for $\ep\to 0$. Furthermore there
exists a function $a(\ep)>0$ with the property $|\log a(\ep)| =
o\left(\frac{1}{\ep}\right),\, \ep\to 0$, such that none of the
operators $H_\ep,H_\ep^{M_1},\ldots H_\ep^{M_r}$ has spectrum in
$[\alpha(\ep)-2a(\ep),\alpha(\ep)[$ or
$]\beta(\ep),\beta(\ep)+2a(\ep)]$.
\een
\end{hyp}

The lattice subset associated to $M_j$ is denoted by $M_{j,\ep} :=
M_j\cap \disk$ and we denote the eigenvalues of $H_\ep$ and of the Dirichlet operators
$H_\ep^{M_j}$ defined in \eqref{HepD} inside the spectral interval  $I_\ep$ and the corresponding real orthonormal systems
of eigenfunctions (these exist because all operators commute with complex conjugation)
\begin{eqnarray}\label{specHepusw}
\spec (H_\ep) \cap I_\ep = \{ \lambda_1,\ldots , \lambda_N\} \,
,&\quad&
u_1,\ldots ,u_N\in \ell^2\left(\disk \right)\\
\F := \Span \{u_1,\ldots u_N\} \nonumber\\
\spec \left(H_\ep^{M_j}\right) \cap I_\ep = \{ \mu_{j,1},\ldots,
\mu_{j,n_j}\} \, ,
&\quad& v_{j,1},\ldots,v_{j,n_j}\in \ell^2_{M_{j,\ep}},\, j\in {\mathcal C} \nonumber\\
\E_j := \Span \{ v_{j,1},\ldots, v_{j,n_j} \} \, , &\quad & \E :=
\bigoplus \E_j\; . \nonumber
\end{eqnarray}
We write 
\begin{equation}\label{valpha}
v_\alpha\quad\text{with}\quad \alpha =(\alpha_1, \alpha_2)\in \mathcal{J}:=\{(j,k)\,|\,j\in\mathcal{C},\, 1\leq k \leq n_j\}
\quad \text{and}\quad j(\alpha):= \alpha_1\, .
\end{equation}
We remark that the number of eigenvalues $N, n_j\, ,\, j\in\mathcal{C}$ with respect to $I_\ep$ as defined in 
\eqref{specHepusw} may depend on $\ep$.

For a fixed spectral interval it is shown in \cite{kleinro4} that the difference
between the exact spectrum and the spectra of Dirichlet
realizations of $H_\ep$ near the different wells is
exponentially small and determined by the Finsler distance between the two nearest neighboring wells.


The first result in this paper generalizes the results of \cite{kleinro3} where we constructed quasimodes for $H_\ep$
(with only one potential well $x_0=0$) in a small neighborhood of $x_0$.
Here we consider the case of several potential wells and construct quasimodes globally on $M_j$ under some additional
assumptions.

\begin{hyp}\label{hypomega}
Let $X_{\tilde{h}_0}$ denote the Hamiltonian vector field of $\tilde{h}_{0}$ defined in \eqref{tildehnull}, 
$F_{t}$ denote the flow of $X_{\tilde{h}_0}$ and set
\begin{equation}\label{Lambdaplus}
\Lambda_{\pm} := \bigl\{ (x,\xi)\in T^*\R^{d}\, |\, F_{t}(x,\xi) \rightarrow (x_j,0)\quad \text{for}
\quad t \rightarrow \mp \infty \bigr\}  \; .
\end{equation}
\ben
\item
For $j\in\mathcal{C}$, let $\O^j\subset \R^d$ containing $x_j$, such that the following
holds:
\ben
\item[(i)] For $\pi : T^*\R^d \rightarrow \R^d$ denoting the bundle projection $\pi (x, \xi) = x$, we have 
\[ \Lambda_+(\Omega^j):=\pi^{-1}(\Omega^j) \cap \Lambda_+ = \bigl\{ (x, \nabla d^j(x)) \in T^*\R^d\, |\, x\in \Omega^j\bigr\} \; . \]
\item[(ii)] $d^j\in \Ce^2(\Omega^j)$.
\item[(iii)] $\pi\bigl(F_t(x,\xi)\bigr) \in \Omega^j$ for all $(x,\xi)\in \pi^{-1}(\Omega^j) \cap \Lambda_+$ and all $t\leq 0$. 
\een
The restriction of $\O^j$ to the lattice is denoted by $\O^j_\ep$.
\item For $j\in\mathcal{C}$ let $M_j$ satisfy Hypothesis \ref{hypIMj} and assume in addition that $M_j\subset \Omega^j$
and that $(a)$ holds for $M_j$ replacing $\Omega^j$.
\een
\end{hyp}

\begin{figure}[htbp]
 \centering
\includegraphics[width=14cm]{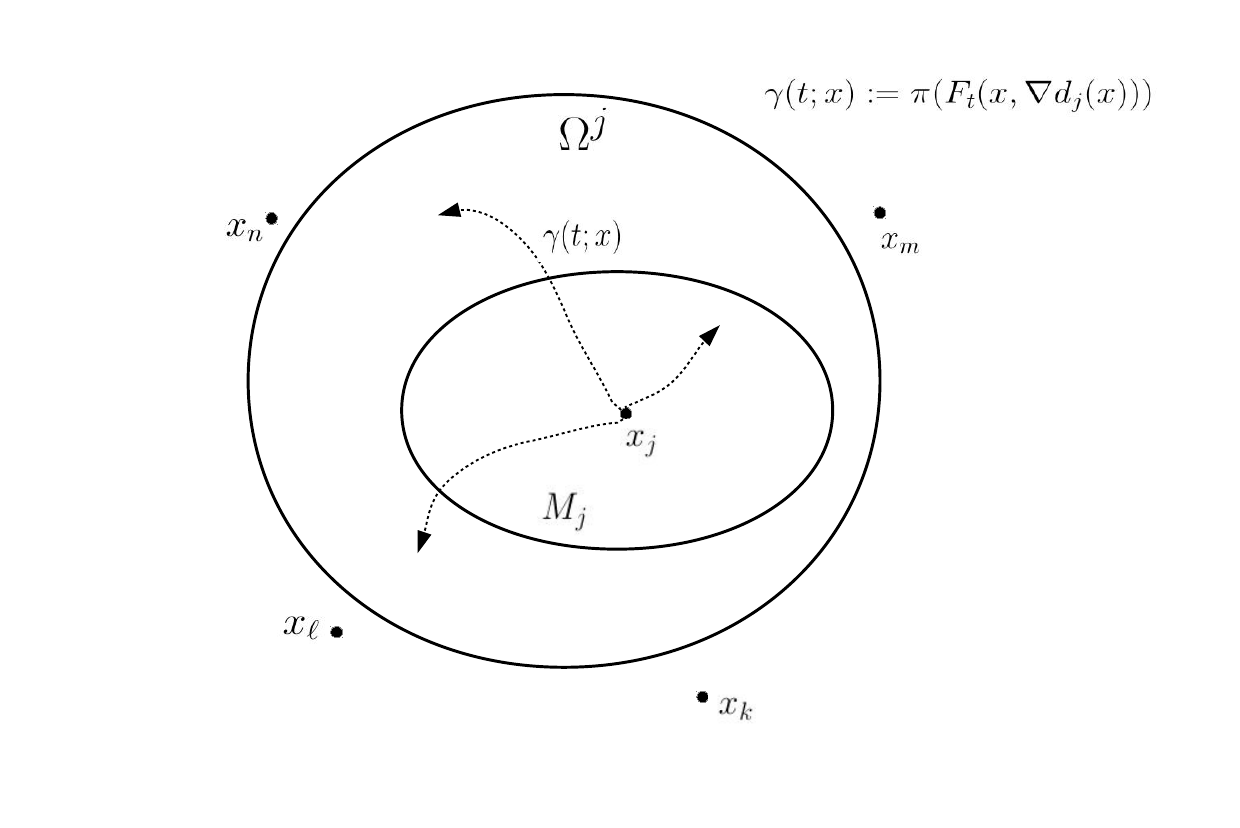}
\caption{The regions $M_j$ and $\Omega^j$ and the projection of the Hamilton flow to $x$-space}
\label{Bild}
\end{figure}

By \cite{kleinro}, Theorem 1.5, the base integral curves of $X_{\tilde{h}_0}$ on
$\R^d\setminus\{x_1,\ldots x_m\}$ with energy $0$ are geodesics with respect to
$d$ and vice versa.
Thus the above hypothesis implies in particular that there is a unique minimal geodesic
between any point in $\Omega^j$ and $x_j$.

Clearly, $\Lambda_+ (\Omega^j)$ is a Lagrange manifold (by (a)(i)) and since the flow $F_t$ preserves $\tilde{h}_0$, we have
$\Lambda_+(\Omega^j)\subset \tilde{h}_0^{-1}(0)$ by \eqref{Lambdaplus}. Thus the eikonal equation 
$\tilde{h}_0(x, \nabla d^j(x)) =0$ holds for $x\in\Omega^j$.
It follows from the construction of the solution of the
eikonal equation in \cite{kleinro3} that in fact $d^j\in \Ce^\infty(\O^j)$.

We recall from Remark 1.4 in \cite{kleinro3} that, locally near $x_j$, the Lagrange manifolds $\Lambda_{\pm}$ are parametrized by
$\pm d^j$ (this follows combining Lemma 1.3 in \cite{kleinro3} with the proof of Theorem 1.5 in \cite{kleinro}). Thus, geometrically speaking,
Hypothesis \ref{hypomega} (a) means that $\Lambda_{\pm}\bigl(\Omega^j\bigr)$ projects diffeomorphically to $\Omega^j$.\\

\begin{def}\label{tilded}
If $\hat{\Omega}^j$ denotes an open neighborhood of $\Omega^j$ given in Hypothesis \ref{hypomega}, let $\chi_j\in \Ce_0^\infty (\R^d)$ 
denote a cut-off function
with $\chi_j(x) = 1$ for $x\in \O^j$ and $\supp \chi_j \subset \hat{\Omega}^j$. Then we set 
\begin{equation}\label{djtilde}
 \tilde{d}^j (x) := \chi_j(x) d^j(x) + \bigl(1-\chi_j(x)\bigr)|x-x_j|\; .
\end{equation}
\end{def}

\begin{theo}\label{theoEjaj}
Let $H_\ep$ be a
Hamilton operator satisfying Hypotheses \ref{hyp1}, \ref{tildevarphihyp} and \ref{Hypo3}.
For $i, j\in\mathcal{C}$, let $\O^j, M_j$ satisfy
Hypothesis \ref{hypomega} and set $S_{ij}:= d(x_i, x_j)$. 
Furthermore we assume that $\ep E^j$ denotes an eigenvalue of $\widehat{H}_{0,q}^{j}$ defined in
\eqref{Hnullhutq} with multiplicity $m_j$. 
Then, for $j\in\mathcal{C}$, there are functions
$b_k^j\in\Ce_0^\infty\left(\R^d\times (0,\ep_0]\right), b_{k,\ell}^{j} \in \Ce_0^\infty(\R^d)\, , k=1,\ldots,m_j\, , \;
\ell\in\frac{\mathbb{Z}}{2}\,,\; \ell \geq -N_j$ for some $N_j\in\N$, supported in $\Omega_j$, such that for all
$M\in \frac{\mathbb{Z}}{2}$ there are
$C_M <\infty$ satisfying
\begin{equation}\label{arep}
\Bigl| b_k^j(z; \ep) - \sum_{\natop{\ell\in\hZ}{-N_j\leq \ell\leq M}}  \ep^\ell
 b_{k,\ell}^{j}(z)\Bigr| \leq C_M \ep^{M+\frac{1}{2}}\, , \quad (z\in\R^d)
\end{equation}
and
real functions $E_k^j(\ep)$ with asymptotic expansion
\begin{equation}\label{Erep}
E_k^j(\ep) \sim E^j + \sum_{s\in\frac{\N^*}{2}} \ep^s E_{k,s}^{j}\; ,
\end{equation}
such that
\ben
\item
for $\tilde{d}^j$ defined in \eqref{djtilde}, the functions 
\begin{equation}\label{hatvjk}
 \hat{v}_{j,k}:= \ep^{-\frac{d}{4}} e^{-\frac{\tilde{d}^j}{\ep}}  b_k^j 
\end{equation}
are almost orthonormal in the sense that 
\begin{equation}\label{orthokont}
\skpR{\hat{v}_{j,k}}{\hat{v}_{i,\ell}} = \delta_{ij}\delta_{k\ell}  + \delta_{ij}O(\ep^\infty) + 
O\bigl( \ep^{-(N_i+N_j+\frac{d}{2})} e^{-\frac{S_{ij}}{\ep}}\bigr)\; .
\end{equation}
\item for any $x_0\in \R^d$, the functions 
\begin{equation}\label{hatvjkep}
\hat{v}_{j,k, x_0}^{\ep} :=\ep^{\frac{d}{2}}r_{{\mathscr G}_{x_0}}\hat{v}_{j,k}
\end{equation}
on the lattice ${\mathscr G}_{x_0} =\disk + x_0$ are approximate eigenfunctions for the operator
$H_\ep$ with respect to the approximate eigenvalues given in \eqref{Erep}, i.e.,
\begin{equation}\label{eigenwgdisk}
e^{\frac{d^j(x)}{\ep}}\bigl(H_{\ep} - \ep E_k^j(\ep)\bigr) \hat{v}_{j,k, x_0}^{\ep}(x) =
O\left(\ep^{\infty}\right)  \,,
\quad (x\in M_j\cap {\mathscr G}_{x_0},\;\ep \to 0)   \, .
\end{equation}
\item for the restricted approximate eigenfunctions $\hat{v}_{j,k, x_0}^{\ep}$, we have
\begin{equation}\label{ortho2}
\Bigl\langle\hat{v}_{j,k, x_0}^{\ep}\, ,\, \hat{v}_{i,\ell, x_0}^{\ep}\Bigr\rangle_{\ell^2(\mathscr{G}_{x_0})}  = 
 \delta_{ij} \delta_{k\ell} + \delta_{ij}O(\ep^\infty) + 
O\bigl( \ep^{-(N_i+N_j+\frac{d}{2})} e^{-\frac{S_{ij}}{\ep}}\bigr)\; .
\end{equation}
\een
\end{theo}

The second result of this paper is the following theorem, in which we compare
the asymptotic eigenfunctions with respect to $\ep E^j$ derived in Theorem \ref{theoEjaj} 
with the exact Dirichlet eigenfunctions associated to a
spectral interval around this eigenvalue of diameter $C_0 \ep^{\frac{3}{2}}$.

\begin{theo}\label{exactasym}
Let $H_\ep$  satisfy Hypotheses \ref{hyp1}, \ref{tildevarphihyp} and \ref{Hypo3}.
For $j\in\mathcal{C}$, let $\O^j, M_j$ satisfy
Hypothesis \ref{hypomega}. 
Furthermore we assume that $\ep E^j$ denotes an eigenvalue of $\widehat{H}_{0,q}^{j}$ defined in
\eqref{Hnullhutq} with multiplicity $m_j$ and we set $I_\ep\bigl(E^j\bigr) = \bigl[\ep E^j - C_0 \ep^{\frac{3}{2}}, \ep E^j + C_0 \ep^{\frac{3}{2}}\bigr]$.
Let $v_{j,k},\, 1\leq k \leq m_j,$ denote real orthonormal eigenfunctions of the Dirichlet-operator $H_\ep^{M_j}$ defined in
\eqref{HepD} with respect to the spectral intervall $I_\ep\bigl(E^j\bigr)$.
Let $\hat{v}^\ep_{j,k}:= \hat{v}_{j,k,0}^\ep, \, 1\leq k \leq m_j,$ be the WKB-functions associated to $\ep E_j$, 
as defined in Theorem \ref{theoEjaj}, \eqref{hatvjkep}
and set $\tilde{v}_{j,k}^{\ep} := \sum_{\ell=1}^{m_j} c_{k,\ell}(\ep) \hat{v}^\ep_{j,\ell}$ where the orthogonal matrix $c_{k,\ell}(\ep)$ is given
in \eqref{cjkfurv} for $\ep>0$ small enoug.

Then for every compact set $K\subset \mathring{M}_j$ and any $N\in\N$ 
\[ \left\| e^{\frac{d^j}{\ep}} \left(v_{j,k}-\tilde{v}_{j,k}\right)
\right\|_{\ell^2(K_\ep)} = O\left(\ep^N\right) \; , \qquad (\ep\to 0) \] 
where $K_\ep:= K\cap \disk$.
\end{theo}

We remark that $c_{k,\ell}(\ep)=0$ if $E^j_k(\ep)$ is not 
asymptotically equal to $\mu_{j,\ell}$ and $\bigl(c_{j,\ell}\bigr)_{1\leq j,\ell \leq m_j}$ can be chosen to be the identity matrix if all $E^j_k(\ep)$ have 
different expansions. \\

The outline of the paper is as follows.

Section \ref{compexas} consists of the proof of Theorem \ref{theoEjaj}, in Section \ref{compexas2} we prove Theorem \ref{exactasym}.
In the Appendix, we restate (adapted versions of) former resluts used in the proofs.

\section{Proof of Theorem \ref{theoEjaj}}\label{compexas}


For $j\in\mathcal{C}$, we use  the system of formal power series  solutions 
$\hat{b}^j_k\in\mathcal{A}_j$ of $H'_{\ep,j,\mathcal{A}}$ (orthonormal with respect to $\langle\, ,\, \rangle_{\mathcal{A}_j}$) with associated
eigenvalues $\ep E^j_k (\ep)$ given in Theorem \ref{theo45}.   

By Borel summation (with respect to $z$), we can find ${\mathscr C}^\infty$-functions
$g^j_{k\ell}$, compactly supported in $C_j\bigl(\O^j-x_j\bigr)$, possessing $\hat{b}^j_{k}$ as Taylor
series at $C_jx_j$. We define a formal asymptotic series in $\O^j$ by
\[ g^j_k(x; \ep) := \sum_{\natop{\ell\in\hZ}{\ell\geq -N}} \ep^\ell U_j g^j_{k\ell}(x-x_j) \; , \]
where $U_j$ is the unitary transformation given in \eqref{unitrans}.
Then by Theorem \ref{theo45} and since $\tilde{d}^j=d^j$ on $\O^j$ and $\tilde{\hat{d}}^j = \hat{d}^j$ on $C_j\bigl(\O^j-x_j\bigr)$, with
$x= x_j + z$,
\begin{equation}\label{WKBmitb}
e^{\frac{\hat{d}^j (z)}{\ep}}(\hat{H}'_{\ep,j} - \ep E^j_k(\ep))
e^{-\frac{\hat{d}^j (z)}{\ep}} \sum_{\natop{\ell\in\hZ}{\ell\geq -N}} \ep^\ell g^j_{k\ell}(z) = 
e^{\frac{d^j(x)}{\ep}}(\hat{H}_{\ep} - \ep E^j_k(\ep))
e^{-\frac{d^j(x)}{\ep}} g^j_{k}(x,\ep)
= r^j_k(x; \ep)
\end{equation}
in the sense of formal power series where $r^j_k(x; \ep) = \sum_{\natop{\ell\in\hZ}{\ell\geq -N}} \ep^\ell
r^j_{k\ell}(x)$ has the property, that each $r^j_{k\ell}$ is in $\Ce_0^\infty(\Omega^j)$ and vanishes to
infinite order at $x=x_j$. It remains to show that it is possible to
modify the functions $U_jg^j_{k,\ell}(\cdot - x_j)$ by (uniquely
determined) functions $c^j_{k\ell}$ vanishing at $x_j$ to
infinite order such that, for the resulting functions
$\tilde{b}^j_{k\ell} := U_jg^j_{k\ell}(\cdot - x_j) - c^j_{k\ell}$, the
formal series
\begin{equation}\label{formsera}
\tilde{b}^j_k(x; \ep):= \sum_{\natop{\ell\geq -N}{\ell\in \Z/2}}\ep^\ell \tilde{b}^j_{k\ell}(x)
\end{equation}
solves, for $x\in M_j$, the equation
\begin{equation}\label{nullinomega3}
e^{\frac{d^j (x)}{\ep}}(\widehat{H}_{\ep} - \ep E^j_k(\ep))
e^{-\frac{d^j (x)}{\ep}} \tilde{b}^j_{k}(x,\ep )=  0
\end{equation}
in the sense of formal power series.
To this end, we have to show that the equation
\begin{equation}\label{diffc}
e^{\frac{\tilde{d}^j (x)}{\ep}}(\widehat{H}_{\ep} - \ep E^j_k(\ep))
e^{-\frac{\tilde{d}^j (x)}{\ep}} c^j_{k}(x,\ep )=
r^j_k(x; \ep)
\end{equation}
has a unique formal power series solution
$c^j_k(x; \ep) = \sum_{\ell\geq -N} \ep^\ell c^j_{k\ell} (x)$ with
coefficients $c^j_{k\ell}\in \mathscr{C}^\infty(\R^d)$ vanishing
to infinite order at $x=x_j$.
By the definition of $\widehat{H}_\ep$ in \eqref{hatHep},
the assumptions in Hypothesis \ref{hyp1} and \eqref{theoev}, the left hand side of \eqref{diffc} is given by
\begin{multline}\label{WKBord}
e^{\frac{\tilde{d}^j(x)}{\ep}} \biggl[ \widehat{T}_\ep + \widehat{V}_\ep
- \ep\,\Bigl(E^j + \sum_{m\in\N^*/2}\ep^m E^j_{km}\Bigr)\biggr]
e^{-\frac{\tilde{d}^j(x)}{\ep}}
\sum_{\natop{\ell\geq -N}{\ell\in \Z/2}}\ep^\ell c^j_{k\ell}(x) \\
=\sum_{\natop{\ell\geq -N}{\ell\in \Z/2}}\ep^\ell \biggl\{
\sum_{\gamma\in\disk}\Bigl[\sum_{m\in\N} \ep^m a_{\gamma}^{(m)}(x)
e^{\frac{1}{\ep}(\tilde{d}^j(x)-\tilde{d}^j(x+\gamma))}c^j_{k\ell}(x+\gamma)\Bigr] \\
 + \sum_{m\in\N /2}\ep^m \left(V_m(x) - \ep
E^j_{km}\right)c^j_{k\ell}(x)\biggr\}
\end{multline}
where we set $E^j_{k0}:= E^j$ and $V_m=0$ for $m\notin \N$.
To get the different orders in $\ep$ of the kinetic term, i.e. the first summand in rhs\eqref{WKBord}, we
expand $\tilde{d}^j$ and $c^j_{k\ell}$ at $x$ and set $\eta
:= \frac{\gamma}{\ep}\in \Z^d$. 
Taylor expansion gives
\begin{equation}\label{entphi}
 \frac{1}{\ep}\left(\tilde{d}^j(x)-\tilde{d}^j(x+\ep\eta)\right) =
-\nabla\tilde{d}^j(x)\cdot\eta - \frac{\ep}{2}
D^2 \tilde{d}^j|_x[\eta]^2 -
 \frac{\ep^2}{2}\int_0^1 (1-t)^2 D^3\tilde{d}^j|_{x+t\ep \eta}[\eta]^3\, dt
\end{equation}
and
\begin{equation}\label{entc}
c^j_{k,\ell}(x+\ep\eta) = c^j_{k,\ell}(x) + \ep\eta\cdot \nabla c^j_{k,\ell}(x) +
\ep^2 \int_0^1(1-t) D^2 c^j_{k,\ell}|_{x+t\ep\eta} [\eta]^2\, dt\; .
\end{equation}
Combining \eqref{entphi} with the
expansion of the exponential function at zero gives
\begin{multline}\label{entephi}
e^{\frac{1}{\ep}(\tilde{d}^j(x)-\tilde{d}^j(x+\ep\eta))} =
e^{-\nabla\tilde{d}^j(x)\cdot\eta}
 \Bigl(
1-\frac{\ep}{2}D^2 \tilde{d}^j|_x[\eta]^2 + \frac{\ep^2}{4}
\bigl(D^2 \tilde{d}^j|_x[\eta]^2 \bigr)^2\!\!\!+ O\left(\ep^4\right)\Bigr)\times\\
 \times \Bigl(1 - \frac{\ep^2}{2}\int_0^1 (1-t)^2 D^3\tilde{d}^j|_{x+t\ep \eta}[\eta]^3
 \, dt + O\left(\ep^4\right)\Bigr) \bigl( 1+O\bigl(\ep^3\bigr)\bigr)\; .
\end{multline}
To lowest order (i.e. to order $\ep^{-N}$), equation \eqref{diffc} is given by
\begin{equation} \label{WKB-N}
\bigl( t_0 (x, \nabla \tilde{d}^j(x)) + V_0(x)\bigr) c^j_{k,-N}(x) = r^j_{k,-N}(x)\; . 
\end{equation}
By the eikonal equation \eqref{eikonal} (which holds in $\O^j$ by Hypothesis \ref{hypomega}), the left hand side of \eqref{WKB-N} 
vanishes in $\Omega^j$. The same argument
applies equation \eqref{diffc} to order $\ep^{-N+\frac{1}{2}}$:
\[
\bigl( t_0 (x, \nabla \tilde{d}^j(x)) + V_0(x)\bigr) c^j_{k,-N+\frac{1}{2}}(x) = r^j_{k,-N+\frac{1}{2}}(x)\; . 
\]
 The first
non-vanishing term in the expansion of \eqref{diffc} arises from the action of the first order part
of the conjugated operator on $c^j_{k,-N}(x)$, which is given by
\begin{multline}\label{firstorder}
 \biggl\{\sum_{\gamma\in\disk}e^{-\frac{1}{\ep}\nabla\tilde{d}^j(x)\cdot\gamma}
 \Bigl[ a_{\gamma}^{(0)}(x)
 \Bigl( \tfrac{1}{\ep} \gamma\cdot\nabla - \tfrac{1}{2\ep} \skp{\gamma}{ D^2 \tilde{d}^j|_x
 \gamma}\Bigr) +
a_{\gamma}^{(1)}(x)\Bigr]  + V_1(x) - E\biggr\} c^j_{k,-N}(x) \\
= r^j_{k,-N+1} \, .
\end{multline}
This equation takes the form
\begin{equation}\label{allform}
\bigl(\mathcal{P}(x,\partial_x) + f(x)\bigr)u(x) = v(x)
\end{equation}
for the differential operator
\begin{equation}\label{mathP}
\mathcal{P} (x,\partial_x) = Z(x)\cdot \nabla \quad\text{for}\quad 
Z(x) := \sum_{\eta\in \Z^d}
a^{(0)}_{\ep\eta}(x) e^{-\eta\cdot \nabla\tilde{d}^j (x)} \eta\; ,
\end{equation}
which is well defined by the exponential decay of $a^{(0)}_{\ep\eta}$ (see Hypothesis
\ref{hyp1}(a)(iv) and since 
\begin{equation}\label{tay8}
\left|\nabla_x \tilde{d}^j|_{\sqrt{\ep}y}\right| = O(\sqrt{\ep}) \quad\text{and}\quad
\partial_x^\alpha \tilde{d}^j|_{\sqrt{\ep}y}=O(1), \;|\alpha|>1\, ,
\qquad (\ep\to 0)\; .
\end{equation}
It follows from the definition of $\tilde{h}_0$ in \eqref{tildehnull} that $Z$ is the velocity field
\begin{equation}\label{velofield}
 Z(x) = \nabla_\xi \tilde{h}_0 \bigl(x, \xi=\nabla \tilde{d}^j(x)\bigr)\; .
\end{equation}
Thus, by Hypothesis \ref{hypomega}, $Z(x)$ is the projection of the  Hamilton field $X_{\tilde{h}_0}$ of $\tilde{h}_0$, evaluated on the
Lagrange manifold $\Lambda_+(\Omega^j)$, onto the configuration space.

We define a cut-off-function $\zeta_j\in\Ce^\infty_0 (\Omega^j)$ such that $\zeta_j(x)=1$ for $x\in M_j$ and we set 
$\tilde{c}^j_{k,-N}:= \zeta_j c^j_{k,-N}$. 
The next and all higher order
equations in \eqref{diffc} result from the action of the first order part of the
conjugated operator given in \eqref{firstorder} on the respective
highest order part of $c^j_k$, which for the $\ell$-th order is
the term $c^j_{k,\ell-1}$. Additionally to the first order
equation, a term is produced by the action of higher order terms of the
conjugated operator on lower order terms of $c^j_k$, which we replace by $\tilde{c}^j_k$. Since
these lower order terms are already determined by the preceding
transport equations, this additional part can be treated as an
additional inhomogeneity of \eqref{allform}. Thus all transport
equations take the form \eqref{allform}
with $f\in{\mathscr C}^\infty\left(\R^d\right)$ and $v\in \Ce^\infty_0(\O^j)$
vanishing to infinite order at $x=x_j$ by the construction of the
formal series \eqref{aj}. In lowest order (i.e. in order $\ep^{-N}$ and $\ep^{-N+\frac{1}{2}}$) the transport equation \eqref{allform}
is homogeneous (i.e. $v=0$).

In order to show that the transport equations can be solved in the space of $\Ce^\infty$-functions vanishing to
infinite order at $x_j$, we first remark that $x_j$ is a singular point of the
vector field $Z$. In fact, $\nabla_\xi t_0(x, 0)=0$ for any $x\in\R^d$ by \eqref{kinen} und $\nabla d^j(x_j)=0$ since, by the eikonal equation together
with the assumptions on $V_0$ in Hypothesis \ref{hyp1}, the distance $d^j$ has a non-degenerate minimum at $x_j$, i.e. 
$D^2 d^j|_{x_j}$ is a positive definite, symmetric matrix.

By \eqref{velofield} together with \eqref{kinen} and \eqref{tildehnull} we get for some $\alpha>0$ 
\begin{equation}\label{LinZ} 
 \skp{x}{D Z|_{x_j}x} = \skp{x}{- 2B(x_j) D^2 d^j|_{x_j} x} \leq - \alpha |x|^2 \, ,\qquad x\in\R^d\, .
\end{equation}
If $\mu (t)$ denotes for $t\in (-\infty ,0]$ the integral curve of $Z$ with $\mu (0) = x_j$, then, by \eqref{LinZ}, $\mu (t)$ 
approaches $x_j$ exponentially fast
(see  e.g. \cite{Wa}). Moreover, since by Hypothesis \ref{hypomega} the integral curves $\mu$ joining any point
in $\Omega^j$ with $x_j$ lie within $\Omega_j$, it is possible to use the method of characteristics and the variation of constants formula
described e.g. in the proof of Proposition 3.5 in Dimassi-Sj\"ostrand \cite{dima}. It follows that \eqref{allform} has a
unique solution $u\in\Ce^\infty(\Omega^j)$ vanishing to infinite order at $x_j$. 

Since in each step the solution was multiplied with the cut-off-function $\zeta_j$, this gives the required compactly supported solution of \eqref{diffc} 
and thus defines $\tilde{b}^j_k$ in
\eqref{formsera} supported in $\Omega^j$ and solving \eqref{nullinomega3} in a neighborhood of $M_j$.

A Borel procedure with respect to $\ep$ gives us a function $b^j_k\in{\mathscr C}^\infty\left(\R^d\times
[0,\ep_0)\right)$ representing the asymptotic sum
$\tilde{b}^j_k(x; \ep)$ given in \eqref{formsera} which is supported in $\O^j$ and solves \eqref{nullinomega3} in a neighborhood
of $M_j$.
Analogously we define a real function $E_j(\ep)$ as an asymptotic
sum
\[
E^j_k(\ep) \sim E^j + \sum_{k\in\frac{\N^*}{2}} \ep^s E^j_{ks}\; .
\]
To make the step from $\widehat{H}_\ep$ acting on
$\Ce_0^\infty\left(\R^d\right)$ to the operator $H_\ep$ acting on
lattice functions in ${\mathcal K}\left(\disk\right)$, we use that
$\mathscr{G}_{x_0}$ is invariant under
the action of $\hat{H}_\ep$.
Thus the restriction to the lattice commutes with $\hat{H}_\ep$
and, using the restriction operator
$r_{\mathscr{G}_{x_0}}$, yields
\eqref{eigenwgdisk}.
We therefore have proven Theorem \ref{theoEjaj} (b) .

For $i=j$, the approximate orthonormality \eqref{orthokont} follows from the orhonormality with respect to $\langle\, ,\,\rangle_{\mathcal{A}_j}$ 
of the expansion $\hat{b}^j_k$ given in \eqref{aj} (Theorem \ref{theo45}). This has to be combined with a standard estimate of Laplace type
($ \int e^{-\frac{x^2}{\ep}} O(x^\infty)\, dx = O(x^\infty)$).
If $i\neq j$, \eqref{orthokont} follows immediately from the estimate
\begin{equation}\label{orthoinichtj}  
\skpR{\hat{v}_{j,k}}{\hat{v}_{i,\ell}} = \ep^{-\frac{d}{2}} \int_{\R^d} b^j_k(x; \ep) b^i_\ell (x; \ep) e^{-\frac{1}{\ep}(d^j(x) + d^i(x))} \, dx \leq 
\ep^{-\frac{d}{2}}e^{-\frac{S_{ij}}{\ep}}\int_{M_j\cap M_i}b^j_k(x; \ep) b^i_\ell (x; \ep)\, dx
\end{equation}
where we used the triangle inequality for $d$. Since $M_k$ is compact for any $k\in \mathcal{C}$, \eqref{orthokont} follows.

The estimate \eqref{ortho2}
for the restricted approximate eigenfunctions
follows for $i=j$ from (a) combined with the general fact that for $a\in\Ce^\infty_0(\R^d,\R)$ and $\varphi\in\Ce^\infty(\R^d, \R)$ satisfying
$\varphi(x_0)=0$ and $D^2\varphi|_{x_0}>0$ for some $x_0\in\R^d$ and $\varphi(x)>0$ for $x\in\supp a\setminus \{x_0\}$ we have
\begin{equation}\label{intsum}
 \ep^{d} \sum_{x\in\disk} a(x) e^{-\frac{\varphi(x)}{\ep}} = \int_{\R^d} a(x) e^{-\frac{\varphi (x)}{\ep}} \, dx + O\bigl(\ep^\infty\bigr)\; .
\end{equation}
This estimate follows from Proposition C.1 combined with the proof of Corollary C.2 in Di Ges\`u \cite{giacomo}. 
In fact, since it is shown in the proof of \cite{giacomo}, Corollary C.2, that 
$\int_{\R^d} \bigl|\partial_y^\alpha a(hy)e^{-\frac{\varphi (hy)}{h^2}}\bigr| \, dy \leq C_\alpha$ independent of $h$
for any $\alpha\in \N^d$, we can use \cite{giacomo}, Proposition C.1, to get
\[  h^{2d} \sum_{y\in (h\Z)^d} a(hy) e^{-\frac{\varphi(hy)}{h^2}} = h^d \int_{\R^d} a(hy) e^{-\frac{\varphi(hy)}{h^2}}\, dy + O\bigl(h^\infty\bigr) \]
and the substitutions $h=\sqrt{\ep}$ and $x = \sqrt{\ep}y$ give the right hand side of \eqref{intsum}.

For $i\neq j$, the arguments are as in \eqref{orthoinichtj}  with summation replacing integration.


\section{Proof of Theorem  \ref{exactasym}}\label{compexas2}

In spirit, Theorem \ref{exactasym} and its proof follow  Helffer-Sj\"ostrand \cite{hesjo}. 
A major technical difference is due to the fact that for second
order differential operators phase functions $\varphi$ of Lipschitz type work in Agmon estimates. In our context, $\varphi \in \Ce^2$ is 
essential (see also \cite{kleinro}). To prove Theorem \ref{exactasym}, it is crucial to modify the phase function used in \cite{kleinro}, and
here we have to differ from \cite{hesjo} to get $\varphi \in\Ce^2$.

First, we need the following result on the exponential decay of eigenfunctions $v_{j,k}$ of the Dirichlet operator 
$H_\ep^{M_j}$ at $x_j, j\in\C,$ associated to a 
spectral interval $I_\ep$.

\begin{prop}\label{weigMj}
Given Hypotheses \ref{hyp1}, \ref{tildevarphihyp}, \ref{Hypo3} and
\ref{hypIMj}, we assume that $I_\ep\subset [0, \ep R]$ for some $R>0$. Then, using the notation in \eqref{specHepusw},
there exists a number $N_0\in\N$ such that for all
$j\in\mathcal C$ and $1\leq k \leq n_j$ and for all
$\ep\in(0,\ep_0]$
\[ \Bigl\| e^{\frac{d^j}{\ep}}v_{j,k} \Bigr\|_{\ell^2(M_{j,\ep})} = O\left(\ep^{-N_0}\right) \; .\]
\end{prop}

\begin{proof} 
We apply the decay estimates for Dirichlet eigenfunctions proven in \cite{kleinro}, Theorem 1.8,
to $H_\ep^{M_j}$, formally replacing $\Sigma$ by $M_j$ and $d^0=d(0,.)$ by $d^j$. To justify this, we remark that, 
without changing $H_\ep^{M_j}$ and $d^j$ on $M_j$, we could modify the potential $V_0$ such that the assumptions 
of \cite{kleinro} are fulfilled up
to the fact that now $V_0$ has (exactly one) minimum at $x_j$ instead of $x_0=0$. One now observes that 
\cite{kleinro}, Theorem 1.8, and
its proof remain valid, if zero is replaced by an arbitrary point $x_j\in\R^d$. 
\end{proof}

In the next proposition, we show that if $\ep E^j$ denotes an eigenvalue of multiplicity $m_j$ of the harmonic oscillator $\hat{H}^j_{0,q}$ 
associated to the Dirichlet operator $H_\ep^{M_j}$,
then $H_\ep^{M_j}$ has $m_j$ eigenvalues inside the interval $[\ep E^j - C_0\ep^{3/2}, \ep E^j + C_0 \ep^{3/2}]$ for some $C_0>0$.

\begin{prop}\label{corbijM1}
Assuming Hypotheses \ref{hyp1}, \ref{tildevarphihyp}, \ref{Hypo3} and \ref{hypIMj}, let $\hat{H}^j_{0,q}$ be the harmonic oscillator defined in
\eqref{Hnullhutq} associated to the Dirichlet operator $H_\ep^{M_j}$ on $M_j\, , j\in\mathcal{C},$ given in \eqref{HepD}. 
Then there exists a bijection $b: \spec (H_\ep^{M_j}) \cap I_\ep \ra
\spec (\hat{H}^j_{0,q})\cap I_\ep$ and a constant $C_0>0$ such that for all
$\ep\in (0,\ep_0]$
\[ |b(\lambda) - \lambda | \leq C_0 \ep^{\frac{3}{2}} \; . \]
\end{prop}

\begin{proof}
In \cite{kleinro4}, Theorem 1.4, we proved that there is a bijection $b: \spec (H_\ep) \cap I_\ep \ra
\spec (H_\ep^{M_j})\cap I_\ep$ such that $b(\lambda) - \lambda = O \bigl(e^{-\frac{C}{\ep}}\bigr)$ for
some $C>0$ and moreover that $\# \spec (H_\ep) \cap I_\ep = O\bigl(\ep^{-N}\bigr)$
for some $N\in\N$ (both in the limit $\ep \to 0$). 
We combine this with a result on the low lying spectra of $H_\ep$ and the associated harmonic oscillator which, for  $A_j, B_j$ in 
\eqref{Hnullhutq} and $\tilde{A}^j:= B_j A_j B_j$, is given by
\begin{equation}\label{harmos}
K = \bigoplus_j K_j\, , \qquad\text{where}\quad K_j (x):= -\Delta + \skp{x}{\tilde{A}^j x} +  V_1(x_j) + t_1(x_j,0)\; .
 \end{equation}
Let $e_k$ denote the $k$-th eigenvalue of $K$ (counting mulitipliciy) and let $E_k(\ep)$ denote the $k$-th
eigenvalue of $H_\ep$ then it is shown in \cite{kleinro2}, Theorem 1.3, that $E_k(\ep) = \ep e_k + O\bigl(\ep^{\frac{6}{5}}\bigr)$
as $\ep\to 0$. 
Together with the results on the approximating eigenvalues given in Theorem \ref{theoEjaj} these results prove Proposition \ref{corbijM1}. 
\end{proof}

By use of Theorem \ref{theoEjaj} and Proposition \ref{corbijM1}, the
next proposition on the distance of eigenspaces follows from
Helffer-Sj\"ostrand \cite{hesjo}, Proposition 1.4. and 2.5, which we recall in the appendix (Proposition \ref{dEFA}).

\begin{prop}\label{E10tildeE10}
For $j\in \mathcal{C}$, let $\mathcal{E}_{j}$ denote the eigenspace of $H_\ep^{M_j}$
for the interval $I_\ep(E^j) = [\ep E^j - C_0\ep^{3/2}, \ep E^j + C_0 \ep^{3/2}]$, where $\ep E^j$ is an eigenvalue of the harmonic
oscillator $\hat{H}^j_{0,q}$ defined in \eqref{Hnullhutq} and $C_0>0$ is some constant. 
Let $\tilde{\mathcal{E}}_{j}$ denote the span of $\{\hat{v}^\ep_{j,1}, \ldots , \hat{v}^\ep_{j, m_j}\}$, where
$\hat{v}^\ep_{j,k}:= \hat{v}^\ep_{j,k,0},\, 1\leq k \leq m_j,$ are the approximate eigenfunctions with respect to 
$\ep E^j$ defined in \eqref{hatvjkep}.
Then, setting $\vec{\dist}(\E, \F) :=  \| \Pi_\E - \Pi_\F \Pi_\E\|$, we get
\begin{equation}\label{distEj}
\vec{\dist}\bigl(\tilde{\E}_{j},\E_{j}\bigr) = \vec{\dist}\bigl(\mathcal{E}_{j},\tilde{\mathcal{E}}_{j}\bigr) =
O\left(\ep^\infty\right)\; . 
\end{equation}
Moreover, the eigenvalues of $H_\ep^{M_j}$
in $I_\ep(E^j)$ are given by $\mu_{j,k} = \ep E^j_k (\ep) +
O\left(\ep^\infty\right),\, 1\leq k \leq m_j$, where $E^j_k(\ep)$ is given in \eqref{Erep}.
\end{prop}

\begin{proof}
By Proposition \ref{corbijM1},  it follows that $\dim \tilde{\E}_j = m_j = \dim \E_j$ and thus by \cite{hesjo}, Proposition 1.4,
$\vec{\dist}\bigl(\tilde{\E}_{j},\E_{j}\bigr) = \vec{\dist}\bigl(\mathcal{E}_{j},\tilde{\mathcal{E}}_{j}\bigr)$.

We estimate $\vec{\dist}\bigl(\tilde{\E}_{j},\E_{j}\bigr)$ using \cite{hesjo}, Proposition 2.5 (see Prop.  \ref{dEFA}).

Theorem  \ref{theoEjaj} and the definition of the Dirichlet operator in \eqref{HepD} gives
\begin{equation}\label{HMjmitrjk} 
H^{M_j}_\ep \hat{v}^\ep_{j,k} = \ep E^j_k(\ep)\id_{M_{j,\ep}} \hat{v}^\ep_{j,k} + \id_{M_{j,\ep}}[H_\ep , \id_{M_{j,\ep}}] \hat{v}^\ep_{j,k} +
 O\bigl(\ep^\infty\bigr) e^{-\frac{d^j}{\ep}}\; . 
\end{equation}
To estimate the norm $\delta$ of the remainder $r_j$ in Proposition \ref{dEFA}, we thus have to analyze  the norm of (cf. the proof of
\cite{kleinro4}, Theorem 1.4) 
\begin{equation}\label{rjkremain}
r_{j,k}(x) := \id_{M_{j,\ep}}[H_\ep,\id_{M_{j,\ep}}] \hat{v}^\ep_{j,k}(x) =
\id_{M_{j,\ep}}(x) \sum_{\gamma\in\disk} a_\gamma (x; \ep) \,\id_{M^c_{j,\ep}}(x+\gamma)\, \hat{v}^\ep_{j,k}(x+\gamma)
\end{equation}
We have
\begin{equation}\label{rjk1}
\|r_{j,k}\|^2_{\ell^2(\disk)} 
\leq \sum_{x\in M_{j,\ep}} \biggl(\sum_{\natop{\gamma\in\disk}{x+\gamma\notin M_{j,\ep}}} 
\Bigl|a_\gamma(x; \ep) \hat{v}^\ep_{j,k}(x+\gamma)\Bigr|\biggr)^2\; .
\end{equation}
To estimate the right hand side of \eqref{rjk1}, we use that, for some $C>0$, we have $\tilde{d}^j(x) > C$ for all $x\notin M_{j}$. 
Thus for $x\in M_{j,\ep}$, we get by \eqref{hatvjk}, 
\begin{equation}\label{rjkl2a}
\sum_{\natop{\gamma\in\disk}{x+\gamma\notin M_{j,\ep}}}\bigl| a_\gamma(x; \ep) \hat{v}^\ep_{j,k}(x+\gamma)\bigr|  =
 \sum_{\natop{\gamma\in\disk}{x+\gamma\notin M_{j,\ep}}} \Bigl|a_\gamma(x; \ep)
 \ep^{\frac{d}{4}}e^{-\frac{\tilde{d}^j(x + \gamma)}{\ep}}b^j_k(x+\gamma)\Bigr|
 \leq e^{-\frac{C}{\ep}} A(x)\; ,
 \end{equation}
where, using the notation $\langle \gamma\rangle_\ep := \sqrt{\ep^2 + |\gamma|^2}$ and
the Cauchy-Schwarz inequality, for any $\eta>0$ and $x\in M_{j,\ep}$
 \begin{align}
 A (x)&:= \sum_{\natop{\gamma\in\disk}{x+\gamma\notin M_{j,\ep}}} \Bigl|a_\gamma(x; \ep)
 b^j_k(x+\gamma)\Bigr| \nonumber\\
  &\leq \biggl(\sum_{\natop{\gamma\in\disk}{x+\gamma\notin  M_{j,\ep}}}
\Bigl| a_\gamma(x; \ep) \langle\gamma\rangle_\ep^{\frac{d+\eta}{2}}\Bigr|^2\biggr)^{\frac{1}{2}}
\biggl( \sum_{\natop{\gamma\in\disk}{x+\gamma\notin  M_{j,\ep}}}
\Bigl|  \langle\gamma\rangle_\ep^{-\frac{d+\eta}{2}} b^j_k(x+\gamma)\Bigr|^2 \biggr)^{\frac{1}{2}}\nonumber\\
&\leq C\biggl( \sum_{\natop{\gamma\in\disk}{x+\gamma\notin  M_{j,\ep}}}
\Bigl| \langle\gamma\rangle_\ep^{-\frac{d+\eta}{2}}
b^j_k(x+\gamma)\Bigl|^2 \biggr)^{\frac{1}{2}}\; ,\qquad x\in M_{j,\ep}\, , \label{rjkA}
\end{align}
where in the second step we choose $\eta$ according to Hypothesis
\ref{hypIMj} (a). Inserting \eqref{rjkl2a} and \eqref{rjkA} in \eqref{rjk1} and
changing the order of summation yields
\begin{align}
\|r_{j,k}\|^2_{\ell^2(\disk)} &\leq e^{-\frac{2C}{\ep}}\tilde{C}\bigl\|
b^j_k\bigr\|^2_{\ell^2_{M_{j,\ep}}}
\sum_{\gamma\in\disk}\langle\gamma\rangle_\ep^{-(d+\eta)}\nonumber\\
&\leq C e^{-\frac{\tilde{C}}{\ep}}\; .\label{rjkabsch}
 \end{align}
Inserting \eqref{rjkabsch} into \eqref{HMjmitrjk} gives the estimate $\delta = O\bigl(\ep^\infty\bigr)$ for $\delta$ as
in Proposition \ref{dEFA} and in addition the statement on the eigenvalues of $H_\ep^{M_j}$.

By the choice of $I_\ep (E^j)$ together with Proposition \ref{corbijM1} it follows that
the constant $a$ in Proposition \ref{dEFA} can be chosen as $O\bigl(\ep^{\frac{3}{2}}\bigr)$. 
Thus we get $\vec{\dist}(\tilde{\E_j}, \E_j) = O\bigl(\ep^\infty\bigr)$. 
\end{proof}

It follows from Proposition \ref{E10tildeE10} that there is an
orthogonal matrix $\left(c_{\ell,k}(\ep)\right)_{1\leq \ell,k \leq m_j}$, such that in $\ell^2(M_{j,\ep})$
\begin{equation}\label{cjkfurv}
v_{j,\ell} = \sum_{k=1}^{m_j} c_{\ell,k} \hat{v}^\ep_{j,k} +
O\left(\ep^\infty\right)\; ,
\end{equation}
where $v_{j, \ell}, \, \ell = 1, \ldots m_j,$ are real normalized eigenfunctions of $H_\ep^{M_j}$ with respect to
$I_\ep(E^j)$ as defind in \eqref{specHepusw}. The matrix $\left(c_{\ell,k}\right)$ can be chosen such that $c_{\ell,k}=0$
if $E^j_k$ is not asymptotically equal to $\mu_{j,\ell}$.
If all $E^j_\ell$ have different expansions, then
$\left(c_{\ell,k}\right)$ may be chosen as identity matrix.

\begin{Lem}\label{d0x<}
Under the assumptions of Hypothesis \ref{hypomega},
for any $x\in\Omega^j$, let $\gamma(t; x)$ denote the base integral
curve of $X_{\tilde{h}_0}$ given by
 \begin{equation}\label{defintku}
 (-\infty,0]\, \ni t\mapsto \pi\circ F_{t}(x,\nabla d^j(x)) =: \gamma(t; x).
\end{equation}
Let $y\in\O^j$ be such that $y\notin \{x_j\}\cup \{\gamma(t; x))\;|\;
-\infty <t \leq  0\}$, then
\begin{equation}\label{ungl} 
d^j(x) < d^j(y) + d(y, x) \; . 
\end{equation}
\end{Lem}

\begin{proof}

By the triangle inequality the statement is true for $\leq$
instead of $<$. The idea of the proof is to show that equality only
may occur, if $y$ lies on the base integral curve of $X_{\tilde{h}_0}$ with
end point $x$.

Let $\eta : [0,1]\ra \O$ be the curve along the segment
$\{0\}\cup \{\gamma(t; x)\;|\;-\infty < t \leq 0\}$, parameterized
such that $\eta(0) = x$ and $\eta(1)= x_j$. Thus, 
by construction and Hypothesis \ref{hypomega}, $\eta$ is a minimal
geodesic between $x_j$ and $x$. In Bao-Chern-Shen \cite{bao}, Thm.
6.3.1, it is shown that minimal geodesics are unique up to
reparametrization.
Equality in \eqref{ungl} would contradict this uniqueness, because
this would mean that there are two different
curves from $x_j$ to $x$, which minimize the curve length and are thus minimizing geodesics.
\end{proof}

By a standard compactness argument, we have the following

\begin{cor}\label{d0<delta}
Let $K_1,\, K_2\subset\O^j$ be compact and assume that $K_2$ is
disjoint from $\hat{K}_1$, the compact union of all minimal
geodesics from all points
of $K_1$ to $x_j$.\\
Then there exists $\delta >0$ such that for all $x\in \hat{K}_1,\,
y\in K_2$
\[ d^j(x) \leq (1-\delta) \left( d^j(y) + d_\ell(y, x) \right) \; . \]
\end{cor}

\begin{proof}[Proof of Theorem \ref{exactasym}]
{\sl Step 1:}\\
In order to simplify the notation, we fix $(j,k)\in \mathcal{J}$ (see \eqref{valpha}) and set 
\begin{equation}\label{rundw} 
r:= (H_\ep^{M_j} - \mu_{j,k}) w  \; ,\qquad w:= v_{j,k}- v_{j,k}' 
\end{equation} 
for $\mu_{j,k}$ denoting the eigenvalue associated to $v_{j,k}$ (see \eqref{specHepusw}), i.e.,
\begin{equation}\label{ewg}
 H_\ep^{M_j}v_{j,k} = \mu_{j,k}v_{j,k}\; .
\end{equation}
We fix some compact set $K\subset \mathring{M}_j$ and write
\begin{align} \label{Ai}
 \id_{K_\ep} \bigl(H_\ep^{M_j} - \mu_{j,k}\bigr) v'_{j,k} &= A_1 + A_2 + A_3\, , \quad\text{where}\\
A_1 &= \sum_{\ell=1}^{m_j} c_{k,\ell} \id_{K_\ep}\bigl(H_\ep - \ep E^j_{\ell}(\ep)\bigr) \hat{v}^\ep_{j,\ell}\; , \nonumber\\
A_2 &= \sum_{\ell=1}^{m_j} c_{k,\ell} \id_{K_\ep} \bigl(\ep E^j_{\ell}(\ep) - \mu_{j,k}\bigr) \hat{v}^\ep_{j,\ell}\; ,\nonumber \\
A_3 &= \sum_{\ell=1}^{m_j} c_{k,\ell} \id_{K_\ep} \bigl[H_\ep, \id_{M_{j,\ep}}\bigr] \hat{v}^\ep_{j,\ell}\; .\nonumber
\end{align}
Theorem \ref{theoEjaj} gives $\bigl\| e^{\frac{d^j}{\ep}} A_1\bigr\|_{\ell^2(K_\ep)} = O \bigl(\ep^\infty\bigr)$.
By Proposition \ref{E10tildeE10} together with the fact  that $c_{k,\ell}=0$ if $E^j_{\ell}$ and $\mu_{j,k}$ are not asymptotically equal
and that $d^j=\tilde{d^j}$ on $M_j$ (cf. \eqref{hatvjkep}) we have the analogue result for $A_2$.
From \cite{kleinro4}, Lemma 5.1, it follows that for any $\delta>0$ and $C>0$, the commutator 
$ \bigl[H_\ep, \id_{M_{j,\ep}}\bigr]$ 
is supported in $\delta M_j:=\{x\, |\, d(x, \partial M_j) < \delta\}$ modulo $O\bigl(e^{-C/\ep}\bigr)$. Since $K\cap \delta M_j = \emptyset$ for some 
$\delta>0$, we get
\begin{equation}\label{abschAi}
 \sum_{j=1}^3 \Bigl\| e^{\frac{d^j}{\ep}} A_j\Bigr\|_{\ell^2(K_\ep)} = O \bigl(\ep^\infty\bigr)\; .
\end{equation}
Inserting \eqref{abschAi} in \eqref{Ai} gives for any compact $K\subset \mathring{M}_j$ 
by \eqref{rundw} and \eqref{ewg}
\begin{equation} \label{gewichtraufK}
 \Bigl\|e^{\frac{d^j}{\ep}} r \Bigr\|_{\ell^2(K_\ep)} = O\bigl(\ep^\infty\bigr)\; .
\end{equation}
Furthermore, by the construction of the WKB function (see \eqref{aj}) and Proposition \ref{weigMj}
\begin{equation}\label{gewichtwaufM}
 \Bigl\| e^{\frac{d^j}{\ep}} w\Bigr\|_{\ell^2(M_{j,\ep})} = O\bigl(\ep^{-N_0}\bigr)
\end{equation}
for some $N_0\in\N$. 
By \eqref{cjkfurv} we have 
\begin{equation}\label{waufM}
 \| w \|_{\ell^2(M_{j,\ep})} = O(\ep^\infty)\; .
\end{equation}
We now claim that for some $N_1\in\N$
\begin{equation}\label{gewichtraufM}
 \Bigl\| e^{\frac{d^j}{\ep}} r \Bigr\|_{\ell^2(M_{j,\ep})} = O(\ep^{-N_1})\; .
\end{equation}
By \eqref{gewichtwaufM}, \eqref{ewg} and \eqref{cjkfurv} it suffices to show
\begin{equation}\label{gewichtHvaufM}
 \Bigl\| e^{\frac{d^j}{\ep}} H_\ep^{M_j} \hat{v}^\ep_{j,k} \Bigr\|_{\ell^2(M_{j,\ep})} = O(\ep^{-N_1})\; .
\end{equation}
To prove \eqref{gewichtHvaufM}, we write (using \eqref{hatvjk})
\begin{equation}\label{3}
 \Bigl( e^{\frac{d^j}{\ep}} H_\ep^{M_j} \hat{v}^\ep_{j,k}\Bigr) (x) = 
 \sum_{\gamma\in\disk} \id_{M_{j,\ep}}(x) \id_{M_{j,\ep}}(x+\gamma) a_\gamma (x; \ep) e^{\frac{1}{\ep}(d^j(x) - d^j(x+\gamma))}\ep^{\frac{d}{4}}
b^j_k(x+\gamma)\; .    
\end{equation}
Now 
\begin{equation}\label{1}
 \sum_{\gamma\in\disk} \id_{M_{j,\ep}}(x+\gamma) \Bigl| a_\gamma (x; \ep) e^{\frac{1}{\ep}(d^j(x) - d^j(x+\gamma)}\Bigr| \leq \sum_{\gamma\in\disk} 
e^{\frac{B|\gamma|}{\ep}} |a_\gamma(x;\ep)|
\leq C
\end{equation}
uniformly for $x\in M_{j,\ep}$, using that $d^j$ is Lipschitz (\cite{kleinro},Theorem 1.6) and \eqref{agammasum}.
Furthermore, by \eqref{arep} in Theorem \ref{theoEjaj}
\begin{equation}\label{2}
 \| b^j_k\|_{\ell^2(M_{j,\ep})} \leq C \ep^{-N_2}
\end{equation}
for some $N_2\in\N$ (cf. the WKB-construction \eqref{aj}). 

Combining \eqref{1} and \eqref{2} with  \eqref{3} gives \eqref{gewichtHvaufM}.\\

{\sl Step 2:}\\
In order to define an appropriate phase function, let $\chi\in{\Ce}^\infty(\R_+,[0,1])$ such that $\chi (r)=0$ for $r\leq
\frac{1}{2}$ and $\chi (r) =1$ for $r\geq 1$. In addition we assume that
$0\leq \chi'(r) \leq \frac{2}{\log 2}$.
For $B>0$ we define
$g: \Omega^j \ra [0,1]$ by
\begin{equation}\label{defg}
g(x):= \chi\left(\frac{d^j(x)}{B\ep}\right)\; ,\qquad x\in\Omega^j
\end{equation}
and set, as in \cite{kleinro},
\begin{equation}\label{DefPhix}
\Phi (x) := d^j(x) - \frac{B\ep}{2}\log \left(\frac{B}{2}\right) -
g(x)\frac{B\ep}{2} \log \left(\frac{2d^j(x)}{B\ep}\right)\; ,\qquad x\in\Omega^j\, .
\end{equation}
To prove Theorem \ref{exactasym}, this phase function is not good enough (by the proof in \cite{kleinro}, it only gives \eqref{gewichtwaufM}).
We need a phase function $\Psi_N$, which improves the decay estimate \eqref{gewichtwaufM} with $\ell^2(M_j)$ replaced by
$\ell^2(K_\ep)$, by a factor $\ep^N$ for any $N\in\N$. Therefore, we consider for $N\in\N$ 
\begin{equation}\label{phasenfunktionN}
\Psi_N(x) := \Phi(x) + \zeta(x) \,N\ep\log\frac{1}{\ep}\; .
\end{equation}
Here $\zeta$ is a well chosen smooth function with $\zeta(x)=1$ for $x\in K$, but $\zeta(x) \leq 0$ in some small neighbourhood of $\partial M_j$.

The first property ($\zeta = 1$ on $K$) gives the required improvement (and can be implemented by using \eqref{gewichtwaufM}), 
the second ($\zeta \leq 0$ near $\partial M_j$) seems to be
unavoidable (since \eqref{gewichtraufM} cannot be improved near the boundary $\partial M_j$).

Furthermore, $\zeta$ should decrease along the outgoing integral curves of the velocity field $Z(x)=\nabla_\xi \tilde{h}_0(x, \nabla d^j(x))$ 
(see \eqref{mathP} and \eqref{velofield}) - 
otherwise it would mess up the positivity properties needed in the following Agmon-type estimates.

We shall define $\zeta$ as the solution of a certain Cauchy problem for the velocity field $Z$.

First notice that, by slightly increasing the compact set $K\subset \mathring{M}_j$, we may without loss of generality assume that $K$ has smooth 
boundary and that
\[ S:=\{ x\in\partial K \, | \, Z(x)\;\text{is tangential to}\; K\} \]
is a hypersurface. This follows from standard transversality arguments (see e.g. Hirsch \cite{hirsch}). For this $K$ we denote by
$\hat{K}$ the union of all minimal geodesics from $x_j$ to points in $K$. Let 
\[ K_1 = \hat{K} \cup ¸\bigl(d^j\bigr)^{-1}([0, \delta]) \; , \]
where $\delta>0$ shall be chosen such that $K_1\subset\mathring{M}_j$. In the following, we shall always assume that
$B\ep<\delta$. For any choice of $B$ (which we shall make in Step 4), this is true for $\ep$ sufficiently small.

Now we define $\zeta$ as the (unique) solution of the Cauchy problem
\begin{equation}\label{cauchy}
 Z\cdot \nabla \zeta = - f\quad\text{in}\; \Omega^j\; , \qquad \zeta|_{\partial K_1} = 1\, ,
\end{equation}
where $f$ is any function in $\Ce^\infty (\Omega^j)$ with $f\geq 0$ and $f(x) = 0$ for $x$ in some neighbourhood of
$K_1$. We remark that in general $\partial K_1$ is not globally smooth (at $S$ and at $\partial \hat{K}\cap (d^j)^{-1}(\delta)$), and 
\eqref{cauchy} is possibly characteristic (at $\partial \hat{K}$). To solve \eqref{cauchy}, first observe that by construction, $Z$ on $\partial K_1$ is 
either tangential
or outgoing. Thus, since $f$ vanishes in a neighbourhood of $K_1$, it follows that $\zeta=1$ in a neighbourhood of $K_1$ and, since the
integral curves of $Z$ finally reach $\partial \Omega_j$ by Hypothesis \ref{hypomega}, \eqref{cauchy} can still be solved by the method of 
characteristics in all of $\Omega_j$. Explicitly,
\begin{equation}\label{lsgcauchy}
 \zeta (x) = 1 - \int_{-\infty}^0 f(\gamma(s;x))\, ds\; ,
\end{equation}
where $\gamma(s;x)$ denotes the integral curve of $\nu(x) = Z(x) \cdot \nabla$ with $\gamma(0;x) = x$. It is clear from \eqref{lsgcauchy} that,
if $f$ is sufficiently large near $\partial M_j$, then $\zeta(x)\leq 0$ in some neighborhood of $\partial M_j$. Summing up, we have
\begin{equation}\label{eigzeta}
 \zeta\in\Ce^\infty(\Omega^j)\, , \qquad \zeta=1\;\text{near} \; K_1\quad \text{and}\quad \zeta\leq 0 \; \text{near}\; \partial M_j\; .
\end{equation}
We remark that for any $B>0$ there is $C>0$ such that for all $x\in M_j$
\begin{equation}\label{ePhied}
e^{\frac{d^j(x)}{\ep}}\frac{1}{C}\left(1+\frac{d^j(x)}{\ep}\right)^{-\frac{B}{2}}\leq
e^{\frac{\Psi_N(x)}{\ep}} \leq
e^{\frac{d^j(x)}{\ep}} \ep^{-N} C\left(1+\frac{d^j(x)}{\ep}\right)^{-\frac{B}{2}}
\end{equation}
which is a direct consequence of Lemma 3.3 in \cite{kleinro}. Furthermore, for $x\in K_1$, the lower bound holds with an additional factor 
$\ep^{-N}$ on lhs\eqref{ePhied} (by \eqref{phasenfunktionN} and \eqref{eigzeta}).\\

{\sl Step 3:}\\
Now the proof follows the proof of Theorem 1.8 in \cite{kleinro}.
We start to give estimates for $\hat{V}_\ep +
V^{\Psi_N}$ where
\begin{equation}\label{Vpsi}
 V^{\Psi_N}(x):=\sum_{\gamma\in M_j'(x)}a_\gamma(x; \ep)
\cosh \left(\tfrac{1}{\ep}(\Psi_N (x+\gamma)-\Psi_N (x))\right)\; , \quad M_j^{'}(x) := \{\gamma\in\disk\,|\, x+\gamma\in M_j\}\; .
\end{equation}
We claim that
\begin{equation}\label{VepundVPhiN}
\hat{V}_\ep(x) + V^{\Psi_N}(x) \geq \begin{cases}  - C_2\,\ep & \qquad\mbox{for}
\quad x\in M_j\cap (d^j)^{-1}([0, B\ep]) \\
\left(\frac{B}{C_0}-C_1\right)\ep &
\qquad\mbox{for}\quad  x\in M_j\cap (d^j)^{-1}([B\ep,\infty))
\end{cases}
\end{equation}
for some $C_0, C_1, C_2>0$ independent of $B$.

To prove \eqref{VepundVPhiN}, we write
\begin{equation}\label{wei5}
\hat{V}_\ep(x) + V^{\Psi_N}(x) = \bigl(\hat{V}_\ep(x) - V_0(x)\bigr) + \bigl(V^{\Psi_N}(x) +
\tilde{t}_0^{M_j}(x, \nabla\Psi_N(x))\bigr)
 + \bigl(V_0(x) -
\tilde{t}_0^{M_j}(x, \nabla\Psi
_N(x))\bigr)
\end{equation}
where, for $M'_j(x)$ defined in \eqref{Vpsi}, we set 
\[ \tilde{t}_0^{M_j}(x,\xi) := - \,\sum_{\gamma\in M_j'(x)}a_\gamma^{(0)}(x) \cosh \Bigl(\frac{\xi\cdot \gamma}{\ep}\Bigr)\; . \]
By Hypothesis \ref{hyp1}
\begin{equation}\label{Vephalbbe}
\hat{V}_\ep(x) - V_0(x) \geq - C_{11} \ep \, , \qquad x\in M_j\; .
\end{equation}
Now we claim that, for some $C_{12}<\infty$ (depending on $N$, but independent of $B$)
\begin{equation}\label{Vnulltnull}
 V_0(x) - \tilde{t}_0^{M_j}(x,\nabla\Psi_N(x)) \geq \begin{cases} 0 & \qquad\mbox{for}
\quad x\in  M_j\cap (d^j)^{-1}([0, B\ep]) \\
\Bigl( \frac{B}{C_0} - C_{12}\Bigr) \ep &
\qquad\mbox{for}\quad  x\in M_j\cap (d^j)^{-1}([B\ep,\infty))\; . \\
\end{cases}
\end{equation}
We first remark that exactly along the lines of the proof of \cite{kleinro}, Theorem 1.8, (3.22), we have 
\begin{equation}\label{Vnulltnullphi}
 V_0(x) - \tilde{t}_0^{M_j}(x,\nabla\Phi(x)) \geq \begin{cases} 0 & \qquad\mbox{for}
\quad x\in  M_j\cap (d^j)^{-1}([0, B\ep]) \\
\frac{B}{C_0} \ep &
\qquad\mbox{for}\quad  x\in M_j\cap (d^j)^{-1}([B\ep,\infty))\; . \\
\end{cases}
\end{equation}
To show \eqref{Vnulltnull}, we consider the following two cases:\\
{\sl Case 1:} $x\in K_1$\\
In this region $\zeta = 1$, thus $\nabla \Psi_N = \nabla \Phi$ and \eqref{Vnulltnull} follows from \eqref{Vnulltnullphi}. \\
{\sl Case 2:} $x\notin K_1$\\
We remark that $\tilde{t}_0^{M_j}(x,\xi)  \leq \tilde{t}_0(x,\xi)$ by Hypothesis \ref{hyp1}(a)(ii). 
Since moreover
\begin{equation}\label{nablaPsiN}
 \nabla \Psi_N (x) = \nabla \Phi (x) + \nabla \zeta (x)\, N\ep \log (\tfrac{1}{\ep}) \; , 
 \end{equation}
Taylor expansion of $\tilde{t}_0$ at the point $\nabla \Phi (x)$ gives by \eqref{Vnulltnullphi}
\begin{equation}\label{case2-4} 
\text{lhs} \eqref{Vnulltnull} \geq V_0(x) - \tilde{t}_0(x, \nabla \Phi (x) ) - I(x) \geq \frac{B}{C_0} \ep - I(x) 
 \end{equation}
where, for $\gamma_x (t) := \nabla \Phi (x) + t \nabla \zeta (x) N \ep \log(\tfrac{1}{\ep})$, the remaining term $I$ is given by 
\begin{equation}\label{I}
 I(x) = N \ep \log (\tfrac{1}{\ep}) \int_0^1 D_\xi \tilde{t}_0(x, \gamma_x (t)) \nabla\zeta (x)  \, dt\; . 
\end{equation}
In order to analyze $I$ we first remark that $|\gamma_x (t) - \nabla d^j(x)| = O( N\ep \log (\tfrac{1}{\ep}))$ and by \eqref{tildehnull}
$|D_\xi \tilde{t}_0|\leq C$ in compacta and thus
\begin{equation}\label{Dxit}
 \bigl| D_\xi \tilde{t}_0 (x, \gamma_x(t)) - D_\xi \tilde{t}_0 (x, \nabla d^j(x))\bigr| = O (N\ep \log (\tfrac{1}{\ep}))\, , \qquad t\in [0,1]\; .
\end{equation}
Inserting \eqref{Dxit} in \eqref{I} gives by the definition of $Z(x)$ as in Step 2
\begin{equation}\label{I2}
 I(x) = N \ep \log(\tfrac{1}{\ep}) Z(x)\cdot \nabla\zeta (x)   + O\Bigl( \bigl(N\ep \log (\tfrac{1}{\ep})\bigr)^2\Bigr) \leq C_{12}\ep 
\end{equation}
for some $C_{12}>0$, where for the second estimate we used that $\zeta$ solves \eqref{cauchy} with $f\geq 0$ and 
$ \bigl(N\ep \log (\tfrac{1}{\ep})\bigr)^2 = o(\ep)$. Inserting \eqref{I2} into \eqref{case2-4} gives \eqref{Vnulltnull}. 

We now claim that for some $C_{13}>0$
\begin{equation}\label{V+t}
 \bigl|V^{\Psi_N}(x) + \tilde{t}_0^{M_j}(x, \nabla\Psi_N(x))\bigr| \leq C_{13}\ep\; .
\end{equation}
Setting, for $M'_j(x)$ defined in \eqref{Vpsi},
\begin{equation}\label{VPsinull}
 V_0^{\Psi_N}(x):=\sum_{\gamma\in M_j'(x)}a^{(0)}_\gamma(x)
\cosh \left(\tfrac{1}{\ep}(\Psi_N (x+\gamma)-\Psi_N (x))\right)\; , 
\end{equation}
we write 
\begin{equation}\label{D1D2}
 V^{\Psi_N}(x) + \tilde{t}_0^{M_j}(x, \nabla\Psi_N(x)) = \bigl( V^{\Psi_N}(x) - V_0^{\Psi_N}(x)\bigr) + \bigr(V_0^{\Psi_N}(x) + 
 \tilde{t}_0^{M_j}(x, \nabla\Psi_N(x))\bigr) =: D_1(x) + D_2(x)
\end{equation}
and analyze $D_1$ and $D_2$ separately. In order to estimate $|D_1(x)|$ we first remark that by \eqref{nablaPsiN} and since $M_j$ and $M'_j(x)$ are
bounded 
\begin{equation}\label{Psi_Nminus}
 \Bigl| \frac{1}{\ep} \bigl(\Psi_N(x) - \Psi_N(x+\gamma)\bigr)\Bigr| \leq C \frac{|\gamma|}{\ep}
\end{equation}
for some $C>0$. Thus, by Hypothesis \ref{hyp1}(a), for some $C>0$ uniformly in $x\in M_j$
\begin{equation}\label{D1ab}
 \bigl|D_1(x)\bigr| \leq \sum_{\gamma\in M'_j(x)} \bigl| \ep a_\gamma^{(1)}(x) + R^{(2)}_\gamma(x; \ep)\bigr| 
\cosh \left(\tfrac{1}{\ep}\bigl(\Psi_N(x) - \Psi_N(x + \gamma)\bigr)\right) \leq C \ep \; . 
\end{equation}
In order to estimate $|D_2(x)|$ we write
\begin{equation}\label{wei3a}
\bigl|D_2(x) \bigr| \leq  
\sum_{\gamma\in M_j'(x)} \bigl| a^{(0)}_\gamma (x)\bigr| \Bigl|\cosh \bigl\{\tfrac{1}{\ep}(\Psi_N (x) -
\Psi_N (x+\gamma))\bigr\}- \cosh\bigl\{\tfrac{1}{\ep}\gamma\nabla\Psi_N (x)\bigr\}\Bigr| \, ,\qquad x\in  M_j \, .
\end{equation}
By the Mean Value Theorem for $\cosh z$ and since $|\sinh x|\leq e^{|x|}$ 
\begin{multline}\label{mittela}
 \Bigl|\cosh \bigl\{\tfrac{1}{\ep}(\Psi_N (x) -
\Psi_N (x+\gamma))\bigr\}- \cosh\bigl\{-\tfrac{1}{\ep}\gamma\nabla\Psi_N (x)\bigr\}\Bigr| \\
\leq 
\sup_{t\in [0,1]} e^{\frac{1}{\ep}| (\Psi_N(x+\gamma) - \Psi_N(x))t - \gamma\nabla\Psi_N(x) (1-t)|}
  \frac{1}{\ep}\Bigl| \Psi_N(x)-\Psi_N (x+\gamma)+\gamma\nabla\Psi_N(x) \Bigr| \; .
\end{multline}
Using \eqref{Psi_Nminus} and $|\gamma \nabla \Psi_N(x)| \leq C|\gamma|$, we get for some $D>0$ uniformly in $x\in M_j$ and $t\in [0,1]$
 \begin{equation}\label{estexpa}
e^{\frac{1}{\ep}|(\Psi_N(x+\gamma) - \Psi_N(x))t +
\gamma\nabla\Phi(x) (1-t)|}  \leq e^{\frac{D}{\ep}|\gamma|}\; .
\end{equation}
Since $\Psi_N\in\Ce^2(M_j)$, we can use second order Taylor expansion to get
\begin{equation}\label{Phi2Taya}
\frac{1}{\ep} \bigl| \Psi_N(x)-\Psi_N(x+\gamma)+\gamma\nabla\Psi_N(x)
\bigr| \leq  \frac{|\gamma|^2}{\ep} \sup_{t\in[0,1]} \bigl| D^2\Psi_N(x+t\gamma)\bigr| \leq C\frac{|\gamma|^2}{\ep}
\end{equation}
uniformly in $x\in M_j$. Inserting \eqref{estexpa} and \eqref{Phi2Taya} into \eqref{mittela} and the result into \eqref{wei3a} gives 
by Hypothesis \ref{hyp1}(a) for some $C, \tilde{C}>0$
\begin{equation}\label{dopababa}
\bigl| D_2(x)\bigr| \leq C \ep \sum_{\gamma\in M'_j(x)} \bigl|a^{(0)}_\gamma (x)\bigr| \Bigl( \frac{|\gamma|}{\ep}\Bigr)^2 
e^{\frac{D}{\ep}|\gamma|} \leq 
\tilde{C} \ep\; .
\end{equation}
\eqref{dopababa} together with \eqref{D1ab} prove \eqref{V+t}. Inserting \eqref{V+t}, \eqref{Vnulltnull} and \eqref{Vephalbbe} into 
\eqref{wei5} proves \eqref{VepundVPhiN}.\\

{\sl Step 4:}\\
In this step, we finish the proof by using Lemma \ref{HepDchi} which is proven in \cite{kleinro} (Lemma 3.2). 
First we remark that $|\nabla \Phi(x)|\leq |\nabla d^j(x)|$ for all $x\in M_j$ (this can be seen as in Step 1 in the proof of Theorem 1.8 in \cite{kleinro}).
Thus by \eqref{nablaPsiN} and the construction of $\zeta$ in Step 2, it follows that $\sup_{x\in M_j}|\nabla \Psi_N(x)|$ is independent
of $B$ and thus the same is true for the constant $C_{T,\Psi_N}$ given in \eqref{FmitF-} (and \eqref{CTPhiab}) by the discussion below 
Lemma \ref{HepDchi}. This allows us to choose $B$ in \eqref{VepundVPhiN} such that
\begin{equation}\label{wahlB} 
\Bigl(\frac{B}{C_0} - C_1\Bigr)\ep -\mu_{j,k} \geq \ep \bigl(C_{T,\Psi_N} + 1\bigr)   \; .
\end{equation}
For
\[   \O_-:= \{ x\in M_j \, |\, \hat{V}_\ep (x) + V^{\Psi_N}(x) - \mu_{j,k} < 0\}\quad\text{and}\quad \O_+ := M_j \setminus \O_- \]
and for $C_{T,\Psi_N}$ as above, we define the functions $F_{\pm} : M_j \ra [0,\infty)$ by
\begin{align}\label{Ffurpsi}
F_+(x) &:= \sqrt{(C_{T,\Psi_N}+1)\ep\id_{\{d^j<B\ep \}}(x) + (\hat{V}_\ep(x) + V^{\Psi_N}(x) -\mu_{j,k})\id_{\O_+}(x)}\\
F_-(x) &:= \sqrt{(C_{T,\Psi_N}+1)\ep\id_{\{d^j<B\ep \}}(x) + (\mu_{j,k} - \hat{V}_\ep(x) -
V^{\Psi_N}(x))\id_{\O_-}(x)}\; .\label{Fminusfurpsi}
\end{align}
Then $F_\pm$ are well defined and (using \eqref{VepundVPhiN} and \eqref{wahlB})
\begin{equation}\label{Feigenschaft2}
F:=F_+ + F_- \geq  \,\sqrt{(C_{T,\Psi_N}+1) \ep} >0\, ,\quad F_-=
O(\sqrt{\ep})\quad \text{and}\quad F_+^2 - F_-^2 = \hat{V}_\ep + V^{\Psi_N} - \mu_{j,k}\; .
\end{equation}
Furthermore, again by \eqref{VepundVPhiN} and \eqref{wahlB},
\begin{equation}\label{suppF-}
\Omega_-\subset \{x\in M_j\,|\, d^j(x)<B\ep\}\cap M_j \quad\text{and thus}\quad \supp F_- = \{d^j\leq B\ep\}\cap M_j\; .
\end{equation}
Now Lemma \ref{HepDchi} yields for $r$ and $w$ defined in \eqref{rundw} and $v = e^{\frac{\Psi_N}{\ep}}w$ the estimate 
\footnote{ Unfortunately, the analog of \eqref{vorlemma} in the proof of Theorem 1.8  in \cite{kleinro} is wrong (the weight
$e^{\frac{\Phi}{\ep}}$ is missing in the analog of the second term on lhs \eqref{vorlemma}). In order to correct this mistake, one should
redefine $F_{\pm}$ in \cite{kleinro} as in the present paper (by the analog of \eqref{Ffurpsi} and \eqref{Fminusfurpsi}) and choose
$B$ as in \eqref{wahlB}, where again the discussion below Lemma \ref{HepDchi} is crucial.}
\begin{equation}\label{vorlemma}
\left\| F e^{\frac{\Psi_N}{\ep}} w\right\|^2_{\ell^2(M_{j,\ep})} - C_{T,\Psi_N} \ep \left\|e^{\frac{\Psi_N}{\ep}}w\right\|^2_{\ell^2(M_{j,\ep})}
\leq 4 \left\| \tfrac{1}{F}e^{\frac{\Psi_N}{\ep}} r \right\|^2_{\ell^2(M_{j,\ep})}
+ 8 \left\| F_- e^{\frac{\Psi_N}{\ep}}w\right\|^2_{\ell^2(M_{j,\ep})}   \,.
\end{equation}
Since $e^{\frac{\Psi_N}{\ep}} = e^{\frac{\Phi}{\ep}} \ep^{-N}$ on $K$ by the definition of $\Psi_N$, we
have for
some $N_0\in \N$ by \eqref{Feigenschaft2} and \eqref{ePhied} (and the discussion below) 
\begin{equation}\label{mitlemma1}
\left\| F e^{\frac{\Psi_N}{\ep}}w\right\|^2_{\ell^2(M_{j,\ep})}- C_{T,\Psi}\ep \left\|e^{\frac{\Psi_N}{\ep}}w\right\|^2_{\ell^2(M_{j,\ep})}\geq 
 \ep \left\|e^{\frac{\Psi_N}{\ep}} w\right\|^2_{\ell^2(K_\ep)} \geq \frac{1}{C}
\ep^{1 + N_0 - 2N}\left\| e^{\frac{d^j}{\ep}}w \right\|_{\ell^2(K_\ep)}
\end{equation}
and by \eqref{suppF-}, \eqref{waufM} and again \eqref{ePhied}
\begin{equation}\label{mitlemma2}
\left\| F_- e^{\frac{\Psi_N}{\ep}}w\right\|^2_{\ell^2(M_{j,\ep})} = \left\| F_- e^{\frac{\Psi_N}{\ep}}
w\right\|^2_{\ell^2(\{d^j<B\ep\}\cap  M_{j,\ep})} \leq C \ep^{1-2N} \left\| w
\right\|^2_{\ell^2(\{d^j<B\ep\}\cap M_{j,\ep})} = O(1)\; ,
\end{equation}
using \eqref{waufM} in the last step.
In order to analyze the remaining term on the right hand side of \eqref{vorlemma}, we introduce a compact set $G\subset M_{j,\ep}$, which is 
chosen such that $\Psi_N(x) \leq \Phi(x)$ for $x\in M_{j,\ep}\setminus G$ (by the construction of $\Psi_N$, this inequality holds on
the neighbourhood of $\partial M_j$ where $\zeta\leq 0$). Then we write
\begin{align}\label{Fgewichtr1}  
 \left\| \tfrac{1}{F}e^{\frac{\Psi_N}{\ep}} r \right\|^2_{\ell^2(M_{j,\ep})}  &= A_1 + A_2\, , \quad \text{where} \\
A_1 = \left\| \tfrac{1}{F}e^{\frac{\Psi_N}{\ep}}r \right\|^2_{\ell^2(G)}\qquad & \text{and}\qquad 
A_2 = \left\| \tfrac{1}{F}e^{\frac{\Psi_N}{\ep}}r \right\|^2_{\ell^2(M_{j,\ep}\setminus G)}\; .\nonumber
\end{align}
We have by  \eqref{Feigenschaft2} 
together with \eqref{ePhied} for some $N_0\in\N$
\begin{equation}\label{Fgewichtr2}
 A_1 \leq C \ep^{-1 + N_0 - 2N} \left\| e^{\frac{d^j}{\ep}}r \right\|^2_{\ell^2(G)} = O(1)\; ,
\end{equation}
where in the last step we used \eqref{gewichtraufK} and the fact that $G$ is compact. For analyzing $A_2$, we use that $\Psi_N\leq \Phi$ on 
$M_{j,\ep}\setminus G$ to get, again by \eqref{Feigenschaft2} 
together with \eqref{ePhied}, for some $N_0, N_1\in\N$
\begin{equation}\label{Fgewichtr3}
 A_2 \leq C \ep^{-1 } \left\| e^{\frac{\Phi}{\ep}}r \right\|^2_{\ell^2(M_j\setminus G)} \leq C \ep^{-1+ N_0} 
\left\| e^{\frac{d^j}{\ep}}r \right\|^2_{\ell^2(M_j\setminus G)} = O(\ep^{-N_1})\; ,
\end{equation}
where we used \eqref{gewichtraufM} in the last step. Inserting \eqref{Fgewichtr2} and \eqref{Fgewichtr3} into
\eqref{Fgewichtr1} and combining the result with \eqref{mitlemma1} and \eqref{mitlemma2} gives by 
\eqref{vorlemma} that for any $N\in N$, there exists $\ep_N$ such that for all $\ep\in (0,\ep_N)$ we have  
\[ \left\| e^{\frac{d^j}{\ep}}w \right\|_{\ell^2(K_\ep)} \leq C \ep^{-1-N_0-N_1 + 2N} \, ,\]
proving Theorem \ref{exactasym}.
\end{proof}

\begin{appendix}

\section{Former results}\label{app}

We restate and adapt some results proven in \cite{kleinro} and \cite{kleinro3}. For details and proofs, we refer to these papers.
We recall that \cite{kleinro} and \cite{kleinro3} only treat the one-well situation (with $x_j=0$). Since all statements of this appendix
are essentially local, the proofs in \cite{kleinro} and \cite{kleinro3} hold unchanged in the present multi-well situation (by use of 
appropriate cut-off functions reducing the multi-well to the one-well situation).

\begin{Def}\label{hypphi}
For $j\in\mathcal{C}$, let $\Omega^j$, $C_j$ and $\hat{d}^j$ be given in Hypothesis \ref{hypomega}, \eqref{unitrans} and \eqref{djhut}.
Let $\tilde{\Omega}^j$ denote an open neighborhood of $C_j\bigl(\O^j - x_j\bigr)$ and let $\chi_j\in \Ce_0^\infty (\R^d)$ denote a cut-off function
with $\chi_j(z) = 1$ for $z\in C_j\bigl(\O^j - x_j\bigr)$ and $\supp \chi_j \subset \tilde{\Omega}^j$. Then we set 
\begin{equation}\label{djtildehut}
 \tilde{\hat{d}^j} (z) := \chi_j(z) \hat{d}^j(z) + (1-\chi_j(z))\,|z|\; .
\end{equation}

and introduce the $\ep$-dependent unitary map
\[
U_{\ep,j}  :
L^2\left(\R^d,dz\right)\rightarrow L^2\left(\R^d,
e^{-2\frac{\tilde{\hat{d}}^j(\sqrt{\ep}y)}{\ep}}dy\right)=:\mathscr{H}_j
\]
by
\begin{equation}\label{unit}
(U_{\ep,j}\,f)(y) =
\ep^{\frac{d}{4}}e^{\frac{\tilde{\hat{d}^j}(\sqrt{\ep}y)}{\ep}}f(\sqrt{\ep}y)
\end{equation}
and set, for $\hat{H}'_{\ep,j}$ as defined in \eqref{Hstrichj},
\begin{equation}\label{defGeps}
\hat{G}_{\ep, j} := \tfrac{1}{\ep}\, U_{\ep,j}\hat{H}'_{\ep,j}U^{-1}_{\ep,j}\; .
\end{equation}
\end{Def}

\begin{prop}\label{TayG}
Let $H_{\ep}$ satisfy Hypothesis \ref{hyp1} and \ref{tildevarphihyp} and let $\hat{H}'_{\ep,j},\, j\in\mathcal{C},$ be
the associated operator defined in  \eqref{Hstrichj}. 
Then the operator $\hat{G}_{\ep,j}$ defined in \eqref{defGeps}
has an expansion
\begin{equation}\label{Gdach}
 \hat{G}_{\ep,j} = \sum_{\tfrac{\N}{2}\ni k\leq N-\frac{1}{2}} \ep^k G^j_k + R^j_N\, ,
 \qquad N\in\hNnull\; .
\end{equation}
Here
\begin{equation}\label{G_kbeiGeps}
G^j_k  =  \left( p^j_{k} + \sum_{|\alpha|=1}^{2k+2} p^j_{k, \alpha}
\partial^\alpha\right)   
\end{equation}
where $p^j_k$ is a polynomial of degree $2k$ which is even (odd) with respect to $y\mapsto -y$
if $2k$ is even (odd), and
$p^j_{k, \alpha}$ is a polynomial of degree $2k + 2 - |\alpha|$ which is even (odd) if $2k -
|\alpha|$ is even (odd).
Moreover, for $\Omega^j$ satisfying Hypothesis \ref{hypomega}, let $\zeta_j\in\Ce_0^\infty(\Omega^j)$ denote a cut-off function and set 
$\zeta_{j,\ep}(y) := \zeta_j (\sqrt{\ep}y)$. Then  there exist constants $C_{N}$ and $\ep_j>0$ such that
for any $\ep\in (0,\ep_j]$ and for any $u,v\in \C[y] \subset \mathscr{H}_{j}$
\begin{equation}\label{restTayG}
 \left| \skpH{u}{\zeta_{j,\ep} R^j_N v} \right| \leq C_{N} \ep^N \sum_{\natop{\alpha\in\N^d}
 {|\alpha|\leq 4N+4}}
 \bigl\|(\,\cdot\, )^\alpha u\bigr\|_{\mathscr{H}_j}
\sum_{\natop{\beta\in\N^d}{|\beta|\leq N}} \bigl\||\, \cdot \, |^{2N+2}\partial^\beta
v\bigr\|_{\mathscr{H}_{j}}\; .
\end{equation}
\end{prop}

\begin{rem}\label{remG_k}
\ben
\item As a map on $\C [y]$, $G^j_k\, , \; k\in\frac{\N}{2}$,
raises the degree of a polynomial by $2k$ and preserves (or
changes) the parity with respect to $y\mapsto -y$ according to the
sign $(-1)^{2k}$. This follows at once from the degree and parity
of the polynomials $p^j_\ell$ in the representation of $G^j_k$.\\
\item The term of order zero more precisely is given by
\begin{equation}\label{G_0}
G^j_0 = \Delta_y \tilde{\hat{d}^j}_0(y) + \sum_{\nu =1}^d
(2(\partial_{y_\nu}\tilde{\hat{d}^j}_0(y))\partial_{y_\nu}) -
\Delta_y + V_1(x_j) + t_1(x_j,0)
\end{equation}
where $\tilde{\hat{d}^j}_0$ is defined similar to \eqref{djtildehut} for $\hat{d^j}_0$ given in \eqref{djentw}. 
The eigenfunctions $h_\alpha$ of $G^j_0$ are products of Hermite polynomials $h_{\alpha_\nu}\in \R[y_\nu]$. Since 
$h_k(-y_\nu)= (-1)^{k }h_k(y_\nu)$ for any $k\in\N, \nu=1, \ldots d$, it follows that $h_{\alpha}$ is even
(respectively odd), if $|\alpha|$ is even (resp. odd).  
\een
\end{rem}

The space $\mathcal{A}_j (\R^d):= \C[[z]][[\ep^{1/2}]]$ of formal Laurent series in $\ep^{1/2}$ with final principal part and  coefficients in the space
$\C[[z]]$ of Laurent series in $z\in \R^d$ with finite principal part is a vector space over the field
$ \mathcal{K} := \C[[\ep^{1/2}]]$ of Laurent series in $\ep^{1/2}$ with final principal part.
There exists a non degenerate sesquilinear form 
$\langle \,\cdot\,,\, \cdot\,\rangle_{\mathcal{A}_j}: \mathcal{A}_j\times \mathcal{A}_j \rightarrow \mathcal{K}$ 
which formally is given by the asymptotic expansion at $z=0$ of 
\[ \langle p,q\rangle_{\mathcal{A}_j} = \ep^{-\frac{d}{2}} \int_{\R^d} \overline{p}(z;\ep) q(z;\ep) e^{-2\frac{d^j(z)}{\ep}}\, dz \]
(see \cite{kleinro3})
and we define
\begin{equation}\label{HaufA}
\left.e^{\frac{\tilde{\hat{d}^j} }{\ep}}\hat{H}'_{\ep,j}e^{-\frac{\tilde{\hat{d}^j}}{\ep}}
\right|_{\mathcal{A}_j}=: H'_{\ep,j, \mathcal{A}}\, ,
\end{equation}
which is symmetric on $\mathcal{A}_j$.

\begin{theo}\label{theo45}
Let $H_{\ep}$ satisfy Hypothesis \ref{hyp1} and \ref{tildevarphihyp} and let $\hat{H}'_{\ep,j},\, j\in \mathcal{C},$ be
the associated operator defined in  \eqref{Hstrichj}. 
Let $\tilde{\hat{d}^j}$ be as given in \eqref{djtildehut} and 
let $E^j$ be an eigenvalue with multiplicity $m_j$ of the harmonic
approximation $G^j_0$ of $\hat{G}_{\ep,j}$ given in \eqref{G_0}.
\ben
\item Then the operator $H'_{\ep,j, \mathcal{A}}$ defined in \eqref{HaufA} has
an orthonormal system (with respect to  $\langle \,\cdot\, ,\, \cdot\, \rangle_{\mathcal{A}_j}$) of $m_j$ eigenfunctions 
$\hat{b}^j_{k},\, k= 1, \ldots m_j,$ in $\mathcal{A}_j$ of
the form 
\begin{equation}\label{aj}
\hat{b}^j_k(z; \ep) := \sum_{\natop{\ell\in\hZ}{\ell\geq -N}}
\ep^\ell \hat{b}^j_{k\ell}(z)
\end{equation}
for some $N\in\Z$, where $\hat{b}^j_{k\ell}(z)$ denote formal power series in $z\in\R^d$
and the lowest order
monomial in $\hat{b}^j_{k\ell}\in\C[[z]]$ is of degree $\max\{-2\ell,0\}$.\\
The associated eigenvalues are
\begin{equation}\label{theoev}
\ep E^j_k(\ep) = \ep \left(E^j + \sum_{\ell\in\hN}\ep^\ell E^j_{k\ell}\right)\, .
\end{equation}
\item We introduce $I_{E^j}:= \left\{ \left. \alpha\in\N^d \,\right|\,
G^j_0 h_\alpha = E^j h_\alpha \right\} =: \{\alpha^1, \ldots ,
\alpha^{m_j}\}$ numbering the $m_j$ Hermite polynomials with eigenvalue $E^j$
for $G^j_0$. 

If $|\alpha|$ is even (resp. odd) for all $\alpha\in I_{E^j}$,
then all half integer (resp. integer) terms in the expansion
\eqref{aj} vanish.
\een
\end{theo}

In order to prove Proposition \ref{exactasym}, we need the following (adapted and reformulated version of a) Lemma proven 
(in a more general setting) in \cite{kleinro} (Lemma 3.2). For the Schr\"odinger operator, this type of result is due to 
Helffer-Sj\"ostrand \cite{hesjo}, see also Dimassi-Sj\"ostrand \cite{dima}.

\begin{Lem}\label{HepDchi}
Assume Hypothesis \ref{hyp1} and for $j\in\mathcal{C}$ let $M_j$ satisfy Hypotheses \ref{hypIMj} and \ref{hypomega}. Let
$\Phi: M_j \ra \R$ be Lipschitz. For $E\geq 0$ fixed, let $F_\pm : M_j \rightarrow [0,\infty)$ be a pair of functions such that
$F(x) := F_+(x) + F_-(x) > 0 $ and
\begin{equation}\label{F+F-bedingung}
F_+^2(x) - F_-^2(x) = \hat{V}_\ep(x) + V^\Phi (x) - E\; , \qquad x\in M_j 
\end{equation}
where 
\begin{equation}\label{Vphi}
 V^{\Phi}(x):=\sum_{\gamma\in M_j'(x)}a_\gamma(x; \ep)
\cosh \left(\tfrac{1}{\ep}(\Phi (x+\gamma)-\Phi (x))\right)\; , \quad M_j^{'}(x) := \bigl\{\gamma\in\disk\,|\, x+\gamma\in M_j\bigr\}\; .
\end{equation}
Then for $v\in\ell^2(M_{j,\ep})$ real-valued, we have for some $C_{T,\Phi}>0$ uniformly with respect to $v$
\begin{equation}\label{FmitF-}
\| F v \|^2_{\ell^2(M_{j,\ep})} \leq 4 \left\| \tfrac{1}{F}\left(e^{\frac{\Phi}{\ep}}(H_\ep - E)
e^{-\frac{\Phi}{\ep}}\right)v\right\|^2_{\ell^2(M_{j,\ep})} + 8 \|F_-v\|^2_{\ell^2(M_{j,\ep})} + C_{T,\Phi} \ep \|v\|^2_{\ell^2(M_{j,\ep})} \, .
\end{equation}
\end{Lem}

It follows at once from the proof of Lemma \ref{HepDchi} in \cite{kleinro}, that the constant $C_{T,\Phi}$ has to fulfill the estimate
\begin{equation}\label{CTPhiab}
-\frac{1}{2}\sum_{x,x+\gamma\in\Sigma}a_\gamma(x, \ep) \cosh \left(\frac{1}{\ep}(\Phi (x)-
\Phi (x+\gamma))\right) (v(x)-v(x+\gamma))^2\geq -C_{T,\Phi} \ep \|v\|^2\; .
\end{equation}
Since $\Phi$ is Lipschitz and $M_j$ is compact, we have $\Phi(x) - \Phi (x+\gamma) \leq \sup_{y\in M_j} |\nabla \Phi (y)|$. It therefore  
follows from Hypothesis \ref{hyp1},(a), (ii),(iv), that such a constant exists and that it 
only depends on $T_\ep$ (with lower bound $C_T$ given in \eqref{Tvonunten}) and $\sup_{y\in M_j} |D\Phi(y)|$).  

For the sake of the reader, we recall Proposition 2.5 from Helffer-Sj{\"o}strand \cite{hesjo}.

\begin{prop}\label{dEFA}
Let $A$ be a self adjoint operator in a Hilbert space $\Hi$ and
$I\subseteq\R$ denote a compact interval. Let
$\mu_1,\ldots,\mu_n\in I$ and $\psi_1,\ldots \psi_n\in\Hi$ be
linearly independent satisfying
\[ A\psi_j = \mu_j\psi_j + r_j  \, ,\]
where $\| r_j\|\leq \delta$.
Let $a>0$ and assume that $\spec (A) \cap
\left((I+B(0,2a))\setminus I\right) = \emptyset$. Denoting by $\E$
the space spanned by $\psi_1,\ldots \psi_n$ and by $\F$ the
eigenspace of $A$ associated with $\spec (A)\cap I$, then we have
\begin{equation}\label{distEFab}
\vec{\dist}(\E,\F) \leq  \frac{ \sqrt{n}\delta}{a\sqrt{\lambda_\Psi^{\min}}}\; ,
\end{equation}
where $\lambda_\Psi^{\min}$ denotes the minimal eigenvalue of the
Gram-matrix $\Psi=\left(\langle{\psi_j},{\psi_k}\rangle_{\Hi}\right)$.
\end{prop}

\end{appendix}

\end{document}